\documentclass[10pt,twocolumn,twoside]{IEEEtran}
\usepackage{cite}
\ifCLASSINFOpdf
\usepackage[pdftex]{graphicx}
\else
\fi
\usepackage{amsmath}
\usepackage{array}
\usepackage{amssymb,amsthm}
\usepackage{booktabs,colortbl}
\usepackage[english]{babel}
\usepackage{color}
\usepackage{eurosym}
\usepackage[font=small]{caption}
\usepackage{subcaption}
\usepackage{dsfont}
\usepackage{relsize}
\usepackage{mathtools}
\usepackage{amssymb}
\usepackage{enumitem} 

\newcommand{\setEder}{\mathcal{N}}

\newcommand{\Dtk}{\delta}

\newcommand{\pvhnb}[1]{{p}_{#1}}
\newcommand{\evh}[1]{e_{#1}}
\newcommand{\evect}{\bold{e}}
\newcommand{\pvect}{\bold{p}}
\newcommand{\tildepvect}{\tilde{\bold{p}}}

\newcommand{\Ivec}[1]{\mathds{1}_{#1}}
\newcommand{\Ivect}[1]{\mathds{1}_{#1}^{\top}}

\newcommand{\nonl}{\renewcommand{\nl}{\let\nl\oldnl}}% Remove line number for one line

\newcommand{\ub}[2]{\overline{#1}_{#2}}
\newcommand{\lb}[2]{\underline{#1}_{#2}}

\newcommand{\rcekagg}{(k+1)}

\newcommand{\setpvhnb}[2]{\mathcal{P}_{#1}(#2)}

\newcommand{\setevh}[2]{\mathcal{E}_{#1}(#2)}
\newcommand{\setevhrcomplete}[1]{\mathcal{R}_{#1}(k+1,\evh{#1}(k))}
\newcommand{\setevhr}[1]{\mathcal{R}_{#1}(k+1)}

\newcommand{\setpvect}[1]{\mathcal{P}(#1)}
\newcommand{\tildesetpvect}[1]{\tilde{\mathcal{P}}(#1)}
\newcommand{\setevect}[1]{\mathcal{E}(#1)}
\newcommand{\setpvhnbagg}[1]{\mathcal{P}_{\textrm{agg}}(#1)}
\newcommand{\tildesetpvhnbagg}[1]{\tilde{\mathcal{P}}_{\textrm{agg}}(#1)}
\newcommand{\setevhagg}[1]{\mathcal{E}_{\textrm{agg}}(#1)}
\newcommand{\setevhraggcomplete}{\mathcal{R}_{\text{agg}}(k+1,E(k))}
\newcommand{\setevhragg}{\mathcal{R}_{\textrm{agg}}(k+1)}
\newcommand{\setevhstaragg}{\mathcal{R}^*_{\textrm{agg}}(k+1)}

\newcommand{\costDispUp}{d_{\text{u}}}
\newcommand{\costDispLow}{d_{\text{l}}}

\newtheorem{ass}{Assumption}
\newtheorem{defi}{Definition}
\newtheorem{thm}{Theorem}

\newtheorem{lemm}{Lemma}

\newtheorem{rema}{Remark}

\newtheorem{example}{Example}

\newcommand*{\QEDA}{\hfill\ensuremath{\blacksquare}}%

\definecolor{timm}{rgb}{.1,.7,.1} %What Timm has changed

\DeclareCaptionFont{tiny}{\tiny}

\DeclareMathOperator*{\argmin}{arg\,min}

\begin{document}
\title{
Towards Optimality Preserving Aggregation for Scheduling Distributed Energy Resources 
}

%, ,  \\ 

\author{Riccardo Remo Appino,~\IEEEmembership{Student Member,~IEEE,}
        Veit Hagenmeyer,~\IEEEmembership{Member,~IEEE,} and \\~Timm Faulwasser,~\IEEEmembership{Member,~IEEE}% <-this % stops a space
\thanks{Manuscript received ...
The authors gratefully acknowledge funding by the German Federal Ministry of Education and Research (BMBF) within the Kopernikus Project ENSURE, %‘New ENergy grid StructURes for the German Energiewende’
BMBF-Support Code 03SFK1A. 
%Moreover, VH and TF  acknowledge support by the Helmholtz Association under the Joint Initiative ``Energy System 2050– A Contribution of the Research Field Energy''.
%TF acknowledges further support from the Baden-W\"urttemberg Stiftung under the Elite Program for Postdocs.
}
\thanks{All authors are with the Institute for Automation and Applied Informatics,
Karlsruhe Institute of Technology, 76344 Eggenstein-Leopoldshafen, Germany; e-mails: \{riccardo.appino, veit.hagenmeyer\}@kit.edu, timm.faulwasser@ieee.org}% <-this % stops a space
}


\markboth{The IEEE Transactions on Control of Network Systems}
{Shell \MakeLowercase{\textit{et al.}}: Bare Demo of IEEEtran.cls for IEEE Journals}



\maketitle

\begin{abstract}
Scheduling the power exchange between a population of heterogeneous distributed energy resources and the corresponding upper-level system is an important control problem in power systems. 
A key challenge is the large number of (partially uncertain) parameters and decision variables that increase the computational burden and that complicate the structured consideration of uncertainties.
Reducing the number of decision variables by means of aggregation can alleviate these issues.
However, despite the frequent use of aggregation for storage, few works in the literature provide formal justification. 
In the present paper, we investigate aggregation of heterogeneous (storage) devices with time-varying power and energy constraints. 
In particular, we propose mild conditions on the constraints of each device guaranteeing the applicability of an aggregated model in scheduling without any loss of optimality in comparison to the complete problem.
We conclude with a discussion of limitations and possible extensions of our findings.   
\end{abstract}


% \fi
\IEEEpeerreviewmaketitle

\section{Introduction}
\label{sec:intro} 

The ongoing paradigm changes in the operation of power systems entail the need for coordinated management and operation of large populations of Distributed Energy Resources (DERs) \cite{Burger17}.
To this end, different explicit and implicit aggregation strategies have been considered in the literature.  
This comprises concepts such as \emph{microgrids}, \emph{virtual power plants}, and \emph{integrated community energy systems}.
Despite differences in aims, level of integration, and operations' management, all of the above consider power exchange across a (virtual/physical) boundary between some kind of lower-level aggregated system and an upper level, which is regulated implicitly via energy markets \cite{Koirala16}. 
In this context, the problem of scheduling the power exchange between a population of heterogenous DERs (or an aggregation thereof) and the corresponding upper-level system is a pivotal problem that has received considerable research attention, cf. the overview \cite{Nosratabadi17} and the references therein. 
A general key challenge is the uncertainty surrounding loads and generators not controllable by the aggregating entity \cite{Appino18a}. 
Hence, scheduling of DERs is usually cast as a stochastic optimization problem.
The additional requirement of computational scalability---i.e. scheduling numerous devices at once---exacerbate the challenge,  as the collective response to the realizations of the uncertainties should consider spatial and temporal correlations.   

However, market regulations require only the total active power exchange with the upper-level grid to be scheduled in advance \cite{Koirala16}; the remaining degrees of freedom (i.e. the  distribution of the power exchange among the individual devices) can be determined in later steps.
This aspect motivates the use of aggregated models of DERs to the end of reducing the complexity of the complete scheduling problem.
The main idea is to calculate \textit{one schedule for each aggregation of DERs} instead of computing \textit{individual schedules for individual DERs}; thus leaving the decision on how to practically operate the individual devices---to enforce the aggregated schedule---to a subsequent moment.
The expected advantages of aggregated models are twofold: computational scalability \textit{and} improved handling of the  uncertainties.\footnote{Note that postponing the decision on the single-device schedules has the additional advantage of improving (statistically) the quality of the overall decision process, because of the additional information available at this latter moment, cf. \cite{King12}.}  

In this context, aggregated scheduling for a population of energy storage systems \cite{Subramanian13}, 
plug-in electric vehicles \cite{Wenzel17,Zhang16,Vandael13,Appino18c}, as well as thermostatically controlled loads \cite{Mathieu15,Hao15} have been recently investigated.
Moreover,  \cite{Xu16,Bernstein15a,Evans18} address populations of heterogenous devices. 
However, aggregating heterogenous DERs is difficult.
In fact, an aggregated model describing a cluster of DERs may induce relaxed constraints, which can lead to solutions violating the original constraints of the individual devices. 
Consequently, applying aggregated models to scheduling calls for guarantees on the feasibility of the solution with respect to the original problem.  
The correctness of reduced aggregated models is largely investigated as a control problem, targeting aspects as frequency response \cite{Villemagne87} or system stability \cite{Rungger18}.
However, the majority of the works dealing with optimization-based control of power systems focuses on convex relaxations of the power flow equations, see \cite{Low14}, disregarding reduced/aggregated models.
Only few works address exactness of generalized aggregated models applied to power systems, e.g. \cite{Evans18,Bernstein15a}. 
In particular, \cite{Evans18} addresses the problem of discharging a cluster of energy-constrained DERs in order to track a pre-computed schedule. 
The authors of \cite{Bernstein15a}, instead, proposed a fully general and theoretically justified aggregation method for the control of distribution systems, followed by an extensive validation study \cite{Reyes15}.

However, so far the following issue remains open: \emph{Is a conceptually simpler interval-based aggregation---as the one presented in \cite{Xu16}---likewise theoretically justified when applied to scheduling problems?}

In the present paper, we investigate scheduling of an aggregation of heterogeneous Energy-Constrained Distributed Energy Resources (EC-DERs), such as storage or flexible loads. 
Similar to \cite{Xu16}, we model both EC-DERs and an aggregation of EC-DERs as so-called ``time-varying batteries" \cite{Mathieu13}, i.e. storage with time-dependent power and energy constraints. 
To the best of the authors' knowledge, a thorough analysis of this aggregation concept has not been conducted yet. 
The contributions of the present paper are:
\begin{itemize}
\item We derive  sufficient conditions that an  aggregated ``time-varying battery" model for a population of  EC-DERs can be used for scheduling without comprising feasibility of the individual constraints of the EC-DERs.
\item We propose a methodology to ensure consistent online tracking of the aggregated schedule. This methodology is based on an appropriate ``dispersion" (that we call \textit{consistent dispersion}), computed via an optimization problem that can be solved point-wise in time.
\end{itemize}

The remainder of the paper is organized as follows: Section \ref{sec:problem_setup} presents the problem; Section \ref{sec:calc_bounds} details a commonly used aggregated model and its limitations; Section \ref{sec:feas_schedule} describes the idea of a consistent dispersion of the aggregated stored energy, ensuring feasibility at the following step; Section \ref{sec:prove_agg_opt} contains the main contribution of the paper: a methodology to compute a consistent dispersion and the proof that such a dispersion always exists; Section \ref{sec:limits_and_adv_of_agg_model} outlines the discussion; Section \ref{sec:conclusion} summarizes our findings.

\section{Scheduling for Integrated Energy Systems} \label{sec:problem_setup}
Subsequently, we introduce the problem analyzed throughout the paper.
First, we introduce the model of a single EC-DER. 
Then, we present a generic scheduling optimization problem.
Finally, we propose a ``time-varying battery" model as aggregation strategy and the corresponding aggregated scheduling.

\subsection{Energy-Constrained Distributed Energy Resources}

A large subset of the technologies referred to as Distributed Energy Resources (DER) is \textit{energy-constrained} \cite{Evans18}. 
Energy Constrained DERs (EC-DERs) are subject to dynamics modeling either the amount of physically stored energy or the integral amount of energy exchanged by the EC-DER over a given time period. 
Examples are batteries and other types of storage, as well as flexible demand.
We recall a quite general model of an EC-DER from \cite{Xu16}.\footnote{Note that the authors of \cite{Xu16} refer to  EC-DERs as ``flexible loads".}

Consider a population of $N$ heterogeneous EC-DERs $\setEder = \{1, \dots, N\}$. 
The $j$-th device is described by a simplified discrete-time model 
\begin{equation}
\evh{j}(k+1) = \evh{j}(k) + \Dtk \cdot \pvhnb{j}(k), \quad \evh{j}(0) = \evh{j}^0,  \label{eq:sd_state_eq_det} \\
\end{equation}
with time index $k \in \mathbb{N}$, sampling period $\Dtk$ and initial condition $\evh{j}^0$. 
The variable $\pvhnb{j}(k) \in \mathbb{R}$ is the (averaged) power output over the $k$-th step and $\evh{j}(k+1) \in \mathbb{R}$ is the energy state at time $k$.
Both $\pvhnb{j}(k)$ and $\evh{j}(k)$ are subject to time-varying constraints
\begin{subequations} 
\begin{align}
\pvhnb{j}(k) &\in  \setpvhnb{j}{k}, \label{eq:sd_pow_limit} \\
\evh{j}(k+1) &\in \setevh{j}{k+1}, \label{eq:sd_en_limit}  
\end{align}
with $\setpvhnb{j}{k}$ and $\setevh{j}{k}$ being closed real intervals, i.e. ${\setpvhnb{j}{k} = [\lb{p}{j}(k),\ub{p}{j}(k)] \subset \mathbb{R}}$ and $\setevh{j}{k+1} = [\lb{e}{j}(k+1), \ub{e}{j}(k+1)] \subset \mathbb{R}$.
\label{eq:ecder_limits}
\end{subequations}
Note that conversion losses are neglected in the model above. 
We will comment on this in Section \ref{sec:system_losses}.

\subsection{Scheduling with Individual Device Models}
\textbf{\begin{figure}
\vspace{-0.2cm}
\centering
\includegraphics[width=0.49\textwidth]{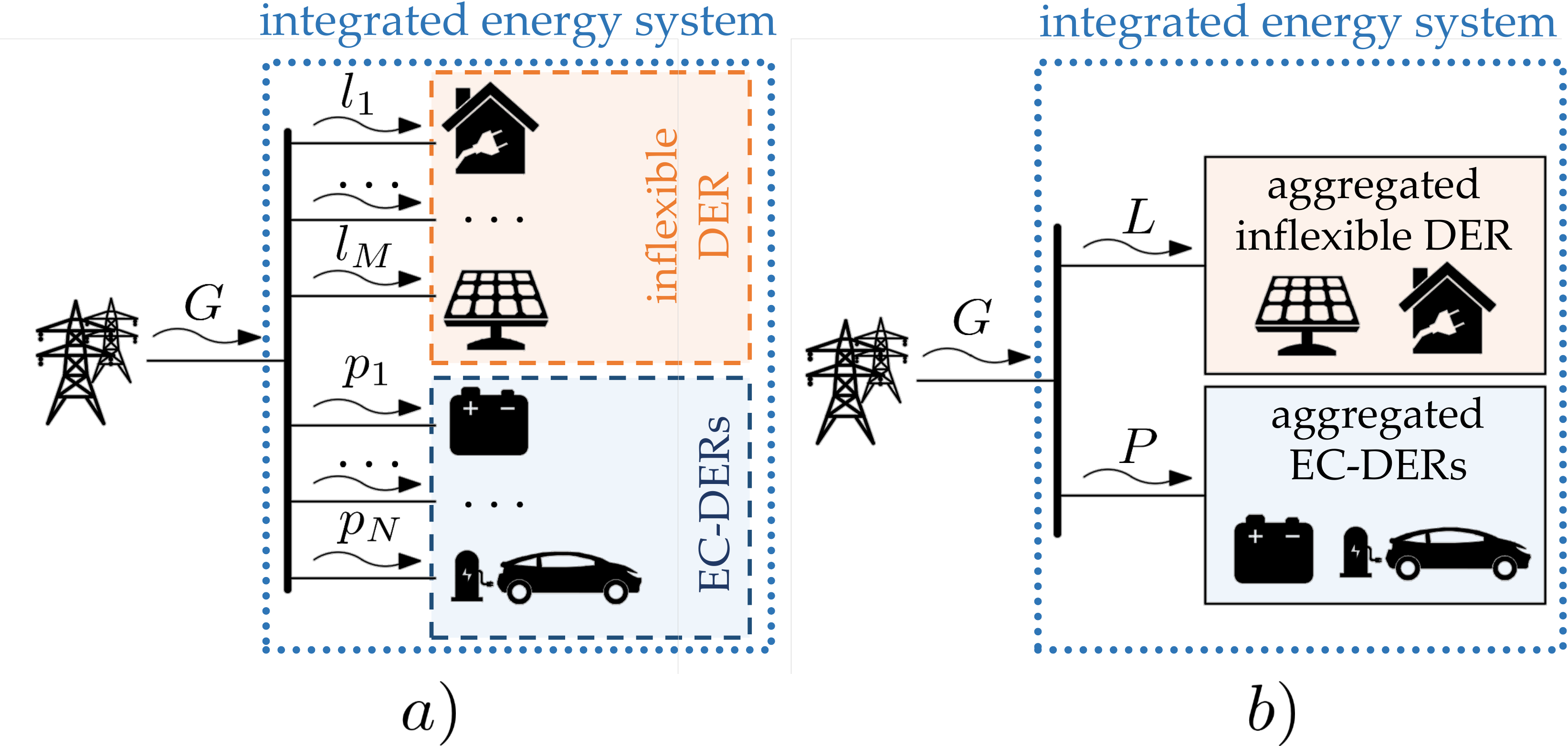}
\caption{Schematic representation of an integrated energy system. \label{fig:local_community} 
}
\vspace{-0.4cm}
\end{figure}}
Consider a microgrid including $M$ inflexible DERs and loads, and $N$ EC-DERs, as illustrated in Fig. \ref{fig:local_community}a.
There, $G(k)$ indicates the aggregated (average)  active power exchange between the internal and the external level at step $k$, $l_{i}(k)$ the aggregated (average) active power output over time step $k$ of the $i$-th inflexible DER/load with ${i\in \mathcal{M}=[1, \hdots, M]}$, and $\pvhnb{j}(k)$ the aggregated (average)  active power output over time step $k$ of the $j$-th EC-DER with $j \in \setEder$.
The connections are assumed to be lossless and such that the system components are able to exchange power mutually without any technical limit. 
The power balance at time $k$ reads
\begin{equation*}
G(k) =  \Ivect{} \pvect(k) - \Ivect{} \bold{l}(k),
\end{equation*}
with $\Ivec{}$ being a vector of ones whose dimension follows from context. 

Consider the problem of minimizing the power exchange $G(k)$ over a time horizon $\mathcal{K} = [0, \hdots, K-1] \subset \mathbb{N}$ with respect to a given cost function $c:~\mathbb{R}\to\mathbb{R}$.
Assuming that the system constraints are given and the values of $l_{i}(k)$ over $\mathcal{K}$ are known for all $i \in \mathcal{M}$, this scheduling problem reads
\begin{subequations}\label{eq:complete_opt_problem}
\begin{align}
\min_{
\begin{subarray}{c}
  \{G(k)\}_{\mathcal{K}}, \\ \{\pvect(k)\}_{\mathcal{K}}, \{\evect(k+1)\}_{\mathcal{K}} % \forall	k\in\mathcal{K}
  \end{subarray}} &\sum_{k \in \mathcal{K}} c\left(G(k)\right)\\
\text{s. t. \,\,} \forall k \in \mathcal{K} \quad& \nonumber \\
  G(k)\phantom{+1,} =& \, \Ivect{} \pvect(k) - \Ivect{} \bold{l}(k), \label{eq:constr_pow_balance}\\ 
 \evect(k+1) =& \, \evect(k) + \Dtk \cdot \pvect(k), \quad \evect(0) = \evect^0,  \label{eq:sd_state_eq_det_vector} \\
 \pvect(k) \phantom{+1,} \in& \, \setpvect{k}, \label{eq:sd_pow_limit_vector}\\
 \evect(k+1) \in&  \, \setevect{k+1}. \label{eq:sd_en_limit_vector} 
% \pvhnb{j}(k) &\in  \setpvhnb{j}{k} \quad \forall j \in \mathcal{N}, \label{eq:sd_pow_limit} \\
% \evh{j}(k+1) &\in \setevh{j}{k+1} \quad \forall j \in \mathcal{N}. \label{eq:sd_en_limit}  
\end{align}
\end{subequations}
Here, we use the vector notation 
\begin{align*}
\bold{l}(k) \doteq [l_{1}(k), \hdots,l_{M}(k)] \in \mathbb{R}^M, \\
\pvect(k) \doteq [\pvhnb{1}(k), \hdots, \pvhnb{N}(k)]  \in \mathbb{R}^N, \\
\evect(k) \doteq [\evh{1}(k), \hdots, \evh{N}(k)] \in \mathbb{R}^N.
\end{align*}
Furthermore, $\{\cdot\}_{\mathcal{K}}$ indicates the sequence of $\cdot$ over $\mathcal{K}$.

The constraints are as follows: Equality
\eqref{eq:constr_pow_balance} is the power balance and \eqref{eq:sd_state_eq_det_vector}-\eqref{eq:sd_en_limit_vector} model the EC-DERs, see \eqref{eq:ecder_limits} and \eqref{eq:sd_state_eq_det}.
The sets $\setpvect{k}$ and $\setevect{k}$ are hyperboxes of $\mathbb{R}^N$ built upon the constraints of the single devices, i.e.  
\begin{subequations}
\label{eq:sd_model_vector} \begin{gather}
\setpvect{k} = \setpvhnb{1}{k} \times ... \times \setpvhnb{N}{k} \subset \mathbb{R}^N, \\ %\label{eq:sd_pow_limit_vector}\\
\setevect{k} = \setevh{1}{k}  \times ... \times \setevh{N}{k} \subset \mathbb{R}^N. %\label{eq:sd_en_limit_vector}
\end{gather}
\end{subequations}

\subsection{Scheduling and Aggregated Models}
\label{sec:scheduling_problem_agg_model}
Upon solving \eqref{eq:complete_opt_problem}---and provided it is feasible---$\{\pvect(k)\}_{k\in\mathcal{K}}$ and $\{\evect(k+1)\}_{k\in\mathcal{K}}$ are obtained. 
However, in terms of coordination with the upper-level grid,  the schedule itself, i.e. $\{G(k)\}_{\mathcal{K}}$, is the only information of crucial interest. 
Moreover, it is also the only decision variable entering the cost function directly.
Thus, the question arises whether it is possible to reduce the number of constraints and decision variables in \eqref{eq:complete_opt_problem} by means of aggregation.

The first step in aggregating DERs  is to  cluster the power outputs of DERs in inflexible ones and energy-constrained ones, see Fig. \ref{fig:local_community}b. 
To this end, consider the function 
\begin{equation}
\label{eq:aggregating_function}
s: \mathbb{R}^{N_x} \rightarrow \mathbb{R}, \quad \bold{x} \mapsto X = \Ivect{} \bold{x},
\end{equation}
which sums---i.e. \emph{aggregates}---the elements of a vector ${\bold{x} \in \mathbb{R}^{N_x}}$.
Formally, its inverse
\begin{equation}
\label{eq:dispersion}
s^{-1}: \mathbb{R} \rightarrow \mathbb{R}^{N_x}, \quad s^{-1}(X) = \{\bold{x} \in \mathbb{R}^{N_x}\,|\, X = \Ivect{} \bold{x} \},
\end{equation}
is the set-valued (pre-image) map. We refer to $s^{-1}(X) \subset  \mathbb{R}^N_x$ as \textit{the set of dispersions of $X$ to $\bold{x}$}.\footnote{This notion is chosen to the end of avoiding confusion with \textit{statistical distributions} of random variables, which are also of interest in scheduling problems.}
The aggregated power outputs and the aggregated energy read 
\begin{equation} \label{eq:agg_powers}
%L(k)\doteq\Ivect{} \bold{l}(k), \quad {P}(k) \doteq \Ivect{} \pvect(k).
L(k)\doteq s(\bold{l}(k)), \quad {P}(k) \doteq s(\pvect(k)),\quad {E}(k) \doteq s(\evect(k)).
\end{equation}
These aggregated variables allow re-writing \eqref{eq:complete_opt_problem} as follows:
\begin{subequations}
\label{eq:complete_opt_problem_with_agg} \begin{align}
\min_{
\begin{subarray}{c}
  \{G(k)\}_{\mathcal{K}}, \\ \{P(k)\}_{\mathcal{K}},\{E(k+1)\}_{\mathcal{K}}, \\ \{\pvect(k)\}_{\mathcal{K}}, \{\evect(k+1)\}_{\mathcal{K}} % \forall	k\in\mathcal{K}
  \end{subarray}} &\sum_{k \in \mathcal{K}} c\left(G(k)\right)\\
  \text{s. t. \,\,} \forall k \in \mathcal{K} \quad& \nonumber \\
G(k)\phantom{+1,} =& \, P(k) - L(k), \label{eq:constr_pow_balance_with_agg}\\ 
 \evect(k+1) =& \, \evect(k) + \Dtk \cdot \pvect(k), \quad \evect(0) = \evect^0, \label{eq:sd_dynamic_vector_with_agg} \\
 \pvect(k)\phantom{+1,} \in& \, \setpvect{k}, \label{eq:sd_pow_limit_vector_with_agg} \\
 \evect(k+1) \in& \,  \setevect{k+1}, \label{eq:sd_en_limit_vector_with_agg} \\
%& {P}(k)\phantom{+1,} =& \, \Ivect{} \pvect(k) \quad &\forall k \in \mathcal{K},  \\
%& {E}(k+1) = \Ivect{} \evect(k+1) \quad &\forall k \in \mathcal{K}, \label{eq:sum_energy_cons_with_agg}
 {P}(k)\phantom{+1,} =& \, s(\pvect(k)), \\
{E}(k+1) =& \, s(\evect(k+1)). \label{eq:sum_energy_cons_with_agg} 
%& \pvhnb{j}(k) \in  \setpvhnb{j}{k} \quad \forall j \in \mathcal{N}, \label{eq:sd_pow_limit} \\
%& \evh{j}(k+1) \in \setevh{j}{k+1} \quad \forall j \in \mathcal{N}, \label{eq:sd_en_limit}  
\end{align}
\end{subequations}
which is equivalent to \eqref{eq:complete_opt_problem} with the additional decision variables $\{P(k)\}_{\mathcal{K}}$ and $\{E(k)\}_{\mathcal{K}}$.

Aggregated scheduling builds upon \eqref{eq:complete_opt_problem_with_agg}, by dropping the constraints on $\pvect(k)$ and $\evect(k)$ and substituting them with constraints on the aggregated variables.
 This yields 
\begin{subequations}\label{eq:aggregated_opt_problem}
\begin{align}
\min_{
\begin{subarray}{c}
   \{G(k)\}_{\mathcal{K}}, \\ \{P(k)\}_{\mathcal{K}},  \{E(k+1)\}_{\mathcal{K}}
  \end{subarray}} &\sum_{k \in \mathcal{K}} c\left(G(k)\right)\\
    \text{s. t. \,\,} \forall k \in \mathcal{K} \quad& \nonumber \\
G(k) \phantom{+1,} =& \, P(k) - L(k), \label{eq:agg_pow_balance}\\ 
 E(k+1) =& \, E(k) + \Dtk \cdot P(k), \quad E(0) = s(\evect^0),\label{eq:agg_state_eq_det} \\
 P(k) \phantom{+1,} \in& \,   \setpvhnbagg{k}, \label{eq:agg_pow_limit} \\
 E(k+1) \in& \, \setevhagg{k+1}. \label{eq:agg_en_limit}
\end{align}
\end{subequations}  
The aggregated problem avoids computing  $\{\pvect(k)\}_{\mathcal{K}}$ and $\{\evect(k)\}_{\mathcal{K}}$ alongside with $\{G(k)\}_{\mathcal{K}}$. 
However, whenever the aggregated constraints sets  $\setpvhnbagg{k}$ and $\setevhagg{k+1}$ do not coincide with the projection of the feasible set of $\pvect(k)$ and $\evect(k)$ from \eqref{eq:complete_opt_problem}, the use of \eqref{eq:aggregated_opt_problem} as a substitute of \eqref{eq:complete_opt_problem} for scheduling purposes may lead to problems.
In the following, we illustrate this via an abstraction.

\subsection{Abstract Problem Statement} \label{sec:abstract_problem}
Conceptually, Problem \eqref{eq:complete_opt_problem} equals
\begin{subequations}\label{eq:abstract_opt_problem}
\begin{align}
\min_{
\begin{subarray}{c}
   \{\bold{x}(k+1)\}_{\mathcal{K}} \end{subarray}} &\sum_{k \in \mathcal{K}}\ell\left(h(\bold{x}(k+1))\right)\\
\text{s.t. \,\,} & \bold{x}(k+1) \in \mathcal{X}(k+1, \bold{x}(k)) \subset \mathbb{R}^{N_x}, 
\end{align}
\end{subequations}  
where the cost function $\ell\circ h:~ \mathbb{R}^{N_x} \rightarrow \mathbb{R}$ is  continuous in $\bold{x}$ and for all $k \in \mathcal{K}$ the constraint set  $\mathcal{X}(k+1, \bold{x}(k))\subset \mathbb{R}^{N_x}$ is non-empty and compact.
Aggregation, as done in \eqref{eq:aggregated_opt_problem}, consists of exploiting the linear projection 
\[
h: \mathbb{R}^ {N_x} \to\mathbb{R}^{N_y}, \quad \bold x \mapsto \bold{y} = h(\bold{x}), \quad \text{with}\quad {N_y < N_x}
\]
to reduce the number of variables.
Applying this to \eqref{eq:abstract_opt_problem} yields
\begin{subequations}\label{eq:abstract_opt_problem_aggregated}
\begin{align}
\min_{
\begin{subarray}{c}
   \{\bold{y}(k+1)\}_{\mathcal{K}} \end{subarray}} &\sum_{k \in \mathcal{K}} \ell\left(\bold{y}(k))\right)\\
\text{s.t. \,\,} & \bold{y}(k+1) \in \mathcal{Y}(k+1, \bold{y}(k)) \subset \mathbb{R}^{N_y}. 
\end{align}
\end{subequations} 
Observe that by construction the cost functions of \eqref{eq:abstract_opt_problem} and \eqref{eq:abstract_opt_problem_aggregated} are equivalent. 
Let $h(\mathcal{X})$ denote the point-wise application of $h$ to $\mathcal{X}$, i.e. 
\begin{equation*}
h(\mathcal{X}) = \left\lbrace \bold{y} = h(\bold{x})\,|\, x \in \mathcal{X} \right\rbrace.
\end{equation*}
 The following lemma is easily obtained. 
\begin{lemm}[Optimality preserving aggregation] \label{lem:equiv}
Consider Problems \eqref{eq:abstract_opt_problem} and \eqref{eq:abstract_opt_problem_aggregated}, let $\bold{x}^\star\in\mathbb{R}^{N_x\cdot K}$ and $\bold{y}^\star\in\mathbb{R}^{N_y\cdot K}$ denote the respective optimal solutions. 
The identity 
\[ 
\bold{y}^\star = h(\bold{x}^\star)
\]
holds for arbitrary continuous choices of $\ell:\mathbb{R}^{N_y} \to \mathbb{R}$ if and only if 
\begin{equation}\label{eq:general_set_equivalence}
\mathcal{Y}(k+1, \bold{y}(k)) \equiv h(\mathcal{X}(k+1, \bold{x}(k))),
\end{equation}
for all $k \in \mathcal{K}$.
\QEDA
\end{lemm}
The proof of this result is straightforward and thus omitted. 
Note that condition \eqref{eq:general_set_equivalence} is not easily enforced via direct calculation of $h(\mathcal{X}(k+1, \bold{x}(k)))$.
In fact, this computation would require the knowledge of the state $\bold{x}$ at each time step, knowledge that is lost during aggregation by construction.
Thus, there is a need of alternative approaches characterizing the set $\mathcal{Y}(k+1, \bold{y}(k))$ such that \eqref{eq:general_set_equivalence} is satisfied.

\section{Calculating Bounds for the Aggregated Model} \label{sec:calc_bounds}
Next, we show that aggregated constraints often used in the context of scheduling problems cannot guarantee the conditions of Lemma \ref{lem:equiv}. 
To this end, we return to the notation of Section \ref{sec:problem_setup}A-C.
Clearly, the construction of the sets $\setpvhnbagg{k}$ and $\setevhagg{k}$ is ambiguous. 
Nevertheless, a common practice (see for example \cite{Xu16}) is to consider intervals of $\mathbb{R}$ given by the (point-wise) projection of the hyperboxes $\setpvect{k}$ and $\setevect{k}$, 
\begin{equation} \label{eq:agg_borders}
\setpvhnbagg{k} = s(\setpvect{k}) , \quad \setevhagg{k} = s(\setevect{k}).
\end{equation}
Practically speaking, this choice leads to
\begin{subequations} 
\begin{equation}
\setpvhnbagg{k} = [\lb{P}{}(k), \ub{P}{}(k)] , \quad \setevhagg{k} = [\lb{E}{}(k), \ub{E}{}(k)],
\end{equation}
with
\begin{equation}
\label{eq:agg_pw_limit_calc_simple}
\lb{P}{}(k) = \sum_{\mathclap{j\in \setEder}} \lb{p}{j}(k), \quad 
\ub{P}{}(k) = \sum_{\mathclap{j\in \setEder}} \ub{p}{j}(k),
\end{equation}
and 
\begin{equation}
\label{eq:agg_en_limit_calc_simple}
\lb{E}{}(k) = \sum_{\mathclap{j\in \setEder}} \lb{e}{j}(k), \quad 
\ub{E}{}(k) = \sum_{\mathclap{j\in \setEder}} \ub{e}{j}(k).
\end{equation}
\end{subequations}
In the sequel, the aggregated constraints sets $\setpvhnbagg{k}$ and $\setevhagg{k}$ are assumed to be obtained via \eqref{eq:agg_borders}.
The underlying ratio stems from interval arithmetics \cite{Jaulin01}. 
Indeed, it can be shown that if $\pvect(k)\in\setpvect{k}$, then 
\[
P(k)= s(\pvect(k)) \in s(\setpvect{k}).
\]
The same holds for the aggregated energy $E(k)$.
Thus, feasibility of $P(k)$ and $E(k)$ in \eqref{eq:sd_pow_limit_vector_with_agg}-\eqref{eq:sum_energy_cons_with_agg} implies their feasibility with respect to  \eqref{eq:agg_pow_limit}-\eqref{eq:agg_en_limit}.
Moreover, it is easy to see that if $e(k)$ and $p(k)$ satisfy \eqref{eq:sd_dynamic_vector_with_agg}, then $P(k)$ and $E(k)$ satisfy \eqref{eq:agg_state_eq_det}.
To summarize, considering \eqref{eq:agg_borders}, the aggregated dynamics \eqref{eq:agg_state_eq_det} and the constraints \eqref{eq:agg_pow_limit}-\eqref{eq:agg_en_limit} do not restrict the feasible set of \eqref{eq:complete_opt_problem_with_agg}.
Rather, they are a relaxation of the projection of the original constraint set in \eqref{eq:complete_opt_problem}.
However, this can cause situations in which the aggregated schedules of $P(k)$ and $E(k)$ cannot be dispersed without violating individual constraints on $\pvect(k)$ and $\evect(k)$.  
Next, we illustrate this issue with a simple example.

\begin{example}[Non-dispersable aggregation] \label{ex:inf_states}
Consider two EC-DERs, i.e. ${\setEder = \{1, 2\}}$, and two subsequent time instants, $k$ and ${k +1}$.
For the sake of simplicity, let the EC-DERs have identical time-invariant power and energy limits, i.e.
$\setpvhnb{1}{k}=\setpvhnb{2}{k}= [\lb{p}{}, \ub{p}{}] \subset \mathbb{R}$, and
$\setevh{1}{k}=\setevh{2}{k}=\setevh{1}{k+1}=\setevh{2}{k+1}= \left[\lb{e}{}, \ub{e}{}\right] \subset \mathbb{R}$.
Furthermore, let $\ub{p}{} = -\lb{p}{}$ with $2\ub{p}{}\Dtk \leq \ub{e}{} - \lb{e}{}$. 

Now, consider the case in which $e_1(k) = \ub{e}{}$ and $e_2(k) = \lb{e}{}$, implying $E(k) = \lb{e}{}+\ub{e}{}$.
The constraints limiting $e_1(k+1)$ and $e_2(k+1)$ are depicted in Fig. \ref{fig:set_example_nointervals}a and Fig. \ref{fig:set_example_nointervals}b respectively.
The blue striped interval represents the feasible states according to the energy constraint; the yellow cone indicates the states admissibly reachable from $e_1(k)$ and $e_2(k)$. 
According to \eqref{eq:agg_borders}, the aggregated constraints are
\begin{subequations}
\label{eq:example_agg_constraints}
\begin{align}
P(k) \in \left[2\lb{p}{}, 2\ub{p}{}\right], \quad
E(k+1) \in \left[2\lb{e}{}, 2\ub{e}{}\right],
\end{align} 
\end{subequations}
see Fig. \ref{fig:set_example_nointervals}c.

For the sake of illustration, consider $E(k+1)=\lb{e}{}+\ub{e}{}+\Dtk\cdot2\ub{p}{}$; the corresponding value of $P(k)=2\ub{p}{}$ follows. 
This choice is feasible for the aggregated model; it satisfies \eqref{eq:example_agg_constraints}.
In Fig. \ref{fig:set_example_nointervals}c this choice of $E(k+1)$ is depicted by an orange circle.
However, the dispersion of $E(k+1)$ and $P(k)$ to individual devices leads to infeasibility on the individual side.

However, a feasible choice of $\pvect(k)$, such that $\pvect(k) = s^{-1}(2\ub{p}{})$ exists, it is:
\begin{align*}
\pvhnb{1}(k) = \ub{p}{} \in [\lb{p}{}, \ub{p}{}], \quad \pvhnb{2}(k) = \ub{p}{} \in [\lb{p}{}, \ub{p}{}].
%&\pvhnb{1}(k_1) = \overline{p}, \quad  \evh{1}(k_2) = \evh{1}(k_1) + \overline{p}\Dtk=\overline{e}+ \overline{p}\Dtk \notin [\lb{e}{}, \ub{e}{}], \\
%&\pvhnb{2}(k_1) = \overline{p}, \quad \evh{2}(k_2) = \evh{2}(k_1) + \overline{p}\Dtk=\underline{e}+ \overline{p}\Dtk \in [\lb{e}{}, \ub{e}{}].
\end{align*}
Yet this implies
\begin{align*}
%\evh{1}(k_2) = \evh{1}(k_1) + \overline{p}\Dtk=\overline{e}+ \overline{p}\Dtk \notin [\lb{e}{}, \ub{e}{}],  \, \evh{2}(k_2) = \evh{2}(k_1) + \overline{p}\Dtk=\underline{e}+ \overline{p}\Dtk \in [\lb{e}{}, \ub{e}{}].
\evh{1}(k+1) = \ub{e}{} + \Dtk \cdot \ub{p}{} \notin [\lb{e}{}, \ub{e}{}],  \,\,\, \evh{2}(k+1) = \lb{e}{}+ \Dtk \cdot \ub{p}{} \in [\lb{e}{}, \ub{e}{}],
\end{align*}
which violates the energy constraint of EC-DER $j=1$.
This can be seen in Fig. \ref{fig:set_example_nointervals}a and \ref{fig:set_example_nointervals}b, where empty red circles indicates energy values which sum up to $E(k+1)$ and are reachable from $e_1(0), e_2(0)$ but not feasible with respect to ${e_1(k+1)\in \left[\lb{e}{}, \ub{e}{}\right]}$.

At the same time, a feasible choice of $\evect(k+1)$ such that $\evect(k+1) = s^{-1}(\lb{e}{}+\ub{e}{}+\Dtk\cdot2\ub{p}{})$ is 
\begin{align*}
\evh{1}(k+1) = \ub{e}{} \in [\lb{e}{}, \ub{e}{}],  \quad \evh{2}(k+1) = \lb{e}{} + \Dtk \cdot 2\ub{p}{} \in [\lb{e}{}, \ub{e}{}].
\end{align*}
This implies $\pvhnb{1}(k) = \frac{\ub{e}{} - \ub{e}{}}{\Dtk} = 
0 \in [\lb{p}{}, \ub{p}{}]$, and
\begin{align*}
%\pvhnb{1}(k) = %\frac{\ub{e}{} - \ub{e}{}}{\Dtk} = 0 \in [\lb{p}{}, \ub{p}{}],~  
\pvhnb{2}(k) = \frac{\lb{e}{} - (\lb{e}{}+\ub{e}{}+2\ub{p}{})}{\Dtk} \notin [\lb{p}{}, \ub{p}{}],
\end{align*}
which violates the power constraint of EC-DER $j=2$.
This is depicted in Fig. \ref{fig:set_example_nointervals}a and \ref{fig:set_example_nointervals}b by red circles, which indicate energy values summing up to $E(k+1)$ that are feasible but not reachable for EC-DER $j=2$.
Hence, the chosen values for $E(k+1)$ and $P(k)$ are feasible for the aggregated model, but cannot be dispersed to $\pvect(k)$ and $\evect(k+1)$. \QEDA
\end{example}

\begin{figure}
\vspace{-0.2cm}
\centering
\includegraphics[width=0.43\textwidth]{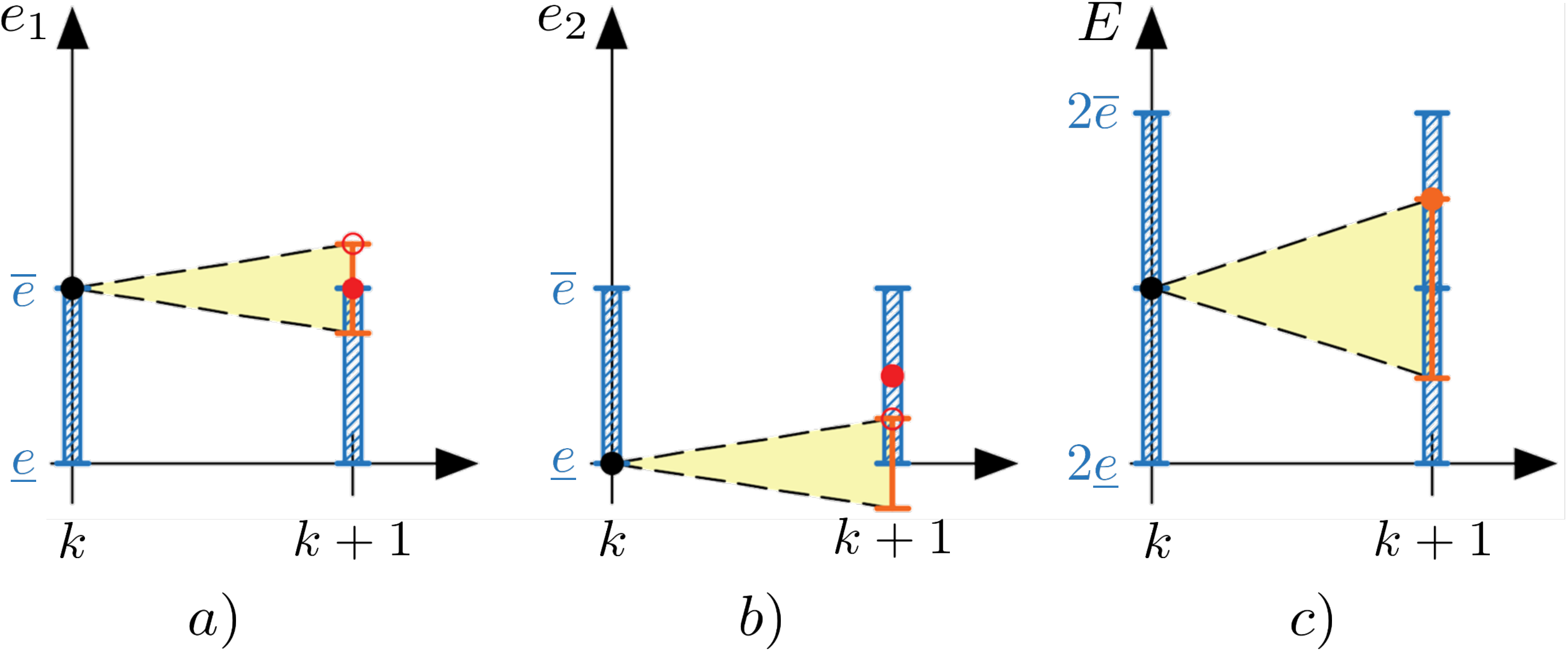}
\caption{Ilustration of Example 1. \label{fig:set_example_nointervals} 
}
\vspace{-0.4cm}
\end{figure}

Example \ref{ex:inf_states} illustrates that an aggregated model with power and energy constraint computed as in \eqref{eq:agg_borders} might lead to infeasibilities in practice.
The time-wise coupling between the variables $\pvhnb{j}(k)$ and $\evh{j}(k)$---and, consequently, between constraints \eqref{eq:ecder_limits}---introduced by the discrete-time dynamics \eqref{eq:sd_state_eq_det} is the main source of these infeasibilities. 
In other words,  we have to account for the constrained reachability properties of the dynamics, ignored in \eqref{eq:agg_borders}. 

Next, we derive an energy constraint summarizing the power and the energy constraints \eqref{eq:ecder_limits}, provided the knowledge of the initial condition $\evect(0)$. 
Given $\evh{j}(k)$, the power constraint \eqref{eq:sd_pow_limit} can be seen as an implicit constraint on $\evh{j}(k+1)$ because it limits the energy state that can be reached at the following step ${k+1}$.
This aspect can be formalized considering the 1-step reachable set of \eqref{eq:sd_state_eq_det}, which is 
\begin{equation}
\label{eq:sd_state_eq_det_int}
%\evh{j}(k+1) \in 
\evh{j}(k) \oplus \Dtk \cdot \setpvhnb{j}{k}.
\end{equation}
Here $\oplus$ denotes the Minkowski sum, and $\Dtk \doteq \setpvhnb{j}{k} = \left\{\Dtk \cdot \pvhnb{j}(k), \forall \pvhnb{j}(k) \in \setpvhnb{j}{k} \right\}$.
However, not all $\evh{j}(k+1) \in \evh{j}(k) \oplus \Dtk \cdot \setpvhnb{j}{k}$ are feasible with respect to the energy constraint at ${k+1}$ \eqref{eq:sd_en_limit}.
Combining \eqref{eq:ecder_limits} and \eqref{eq:sd_state_eq_det} leads to 
\begin{equation}
\label{eq:complete_power_limit_energy_focus}
\evh{j}(k+1) \in \setevhrcomplete{j},
\end{equation}
where 
\begin{equation}
\label{eq:complete_power_limit_energy_focus_constraint_set}
\setevhrcomplete{j} \doteq \setevh{j}{k+1} \cap \left( \evh{j}(k) \oplus \Dtk \cdot \setpvhnb{j}{k}\right).
\end{equation}
The set $\setevhrcomplete{j}$ is the interval of feasible and reachable $\evh{j}(k+1)$ given $\evh{j}(k)$.
For the sake of readability, we will omit the explicit dependency of the interval $\setevhrcomplete{j}$ on $\evh{j}(k)$ in the sequel. 
Evidently, one needs to avoid cases in which the intersection in \eqref{eq:complete_power_limit_energy_focus_constraint_set} is the empty set. 
To this end, we consider the following assumption. 
\begin{ass}[Consistency of constraints] \label{req:meaningful_device_constraints_part1}
For all $j \in \mathcal{N}$ EC-DERs and all $k \in \mathcal{K}$, let the sampling time $\delta$ and the intervals $\setpvhnb{j}{k}$, $\setevh{j}{k}$ and $\setevh{j}{k+1}$ be such that for all $\evh{j}(k) \in \setevh{j}{k}$
\begin{align} \label{eq:ass1}
\setevh{j}{k+1} \cap \left( \evh{j}(k) \oplus \Dtk \cdot \setpvhnb{j}{k}\right) \neq \varnothing
%\setevh{j}{k+1} &\supseteq \bigcap_{\evh{j}(k) \in \setevh{j}{k}} \left( \evh{j}(k) + \Dtk \cdot \setpvhnb{j}{k} \right).
\end{align}
holds.
\end{ass}
This condition can be regarded as controlled forward invariance of the entire (time-varying) state constraint set $\setevh{j}{k}, k \in \mathcal{K}$.
Practically speaking, it can be satisfied by reducing the state constraints---i.e. if there are no feasible $\evh{j}(k+1)$ reachable from $\evh{j}(k)$, then this $\evh{j}(k)$ should be excluded from $\setevh{j}{k}$---or by enlarging the set $  \Dtk \cdot\setpvhnb{j}{k}$.

Under Assumption \ref{req:meaningful_device_constraints_part1}, the bounds of 
\begin{subequations}\label{eq:complete_power_limit_energy_focus_boundaries}
\begin{equation}
\setevhr{j}=\left[\lb{e}{j}^r(k+1), \ub{e}{j}^r(k+1)\right]
\end{equation}
are given by
\begin{align}
&\lb{e}{j}^r(k+1) = \max\left(\lb{e}{j}(k+1),\evh{j}(k) + \Dtk \cdot \lb{p}{j}(k)\right),\\
&\ub{e}{j}^r(k+1) = \min\left(\ub{e}{j}(k+1),\evh{j}(k) + \Dtk \cdot \ub{p}{j}(k)\right);
\end{align}
\end{subequations}
which follows from standard tools of interval arithemetics \cite{Jaulin01}.
Going back to Example \ref{ex:inf_states}, a sketch of $\setevhr{1}$ and $\setevhr{2}$ is shown in Fig. \ref{fig:set_example_intervals}a and Fig. \ref{fig:set_example_intervals}b.
Therein, these intervals are depicted by green crossed intervals.

\begin{figure}
\centering
\includegraphics[width=0.45\textwidth]{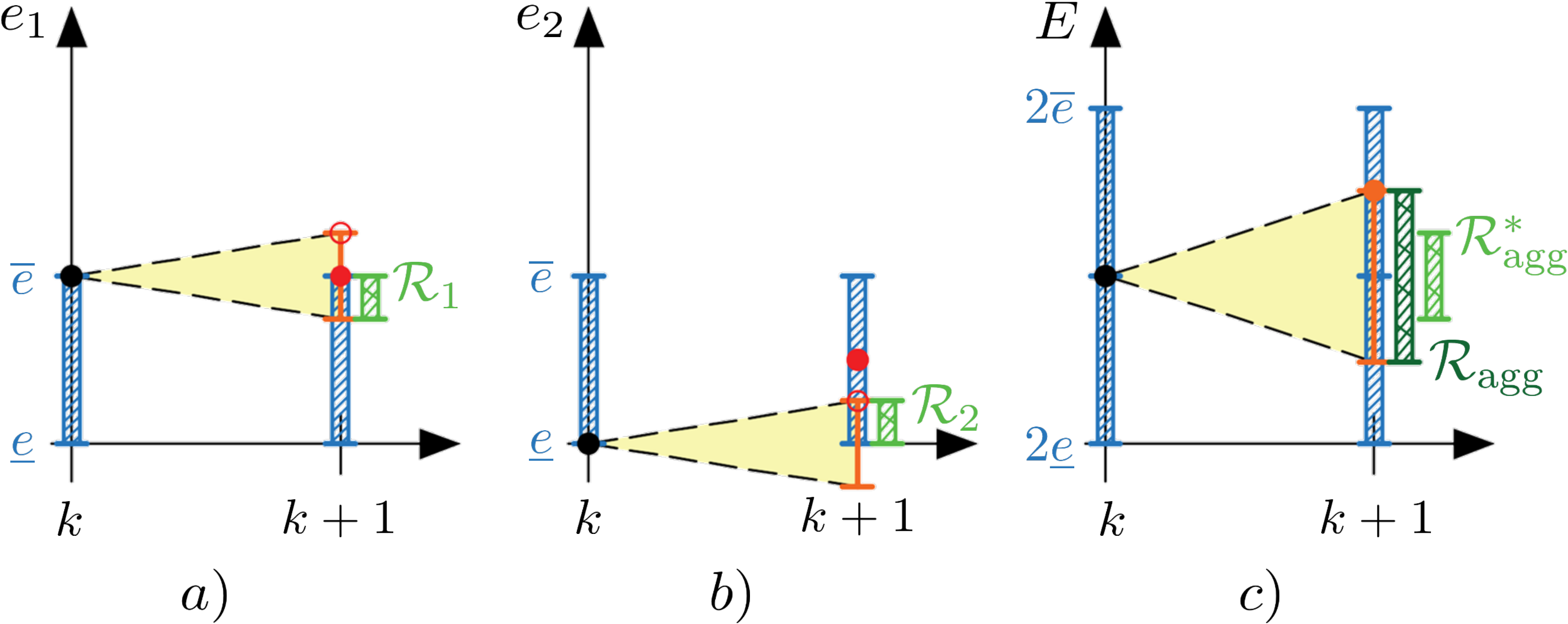}
\caption{\label{fig:set_example_intervals} Illustration of the energy constraints \eqref{eq:complete_power_limit_energy_focus} and \eqref{eq:complete_power_limit_energy_focus_agg} in Example \ref{ex:inf_states}.
}
\vspace{-0.4cm}
\end{figure}

An energy constraint summarizing the power and the energy constraints can be elaborated also for the aggregated states.
This model requires a joint power and energy constraint (per time step), i.e. 
\begin{equation}
\label{eq:complete_power_limit_energy_focus_agg}
E(k+1) \in \setevhraggcomplete,
\end{equation}
with
\begin{equation*}
\setevhraggcomplete \doteq \setevhagg{k+1} \cap \left( E(k) \oplus \Dtk \cdot \setpvhnbagg{k}\right).
\end{equation*}
Henceforth, we simplify the notation by dropping in the following the explicit dependency of $\setevhraggcomplete$ on $E(k)$. 
Similar to \eqref{eq:complete_power_limit_energy_focus_boundaries}, the bounds of
\begin{subequations}\label{eq:aggregated_bounds}
\begin{equation}
\setevhragg = \left[\lb{E}{}^r\rcekagg, \ub{E}{}^r\rcekagg \right]
\end{equation}
are  
\begin{align}
&\lb{E}{}^r(k+1) = \max\left(\lb{E}{}(k+1),E(k) + \Dtk \cdot \lb{P}{}(k)\right),\\
&\ub{E}{}^r(k+1) = \min\left(\ub{E}{}(k+1),E(k) + \Dtk \cdot \ub{P}{}(k)\right).
\end{align}
\end{subequations}

However, given $\evect(k)$, the actual aggregation of the feasible and reachable energy intervals of each EC-DER $\setevhr{j}$ gives
\begin{equation} \label{eq:direct_computation_of_agg}
\setevhstaragg \doteq s(\setevhr{}) % = \left[ \sum_{\mathclap{j\in \setEder}}\lb{e}{j}^r(k+1),\sum_{\mathclap{j\in \setEder}} \ub{e}{j}^r(k+1)\right],
\end{equation}
where $\setevhr{}$ is an hyperbox of $\mathbb{R}^N$ built upon the combined constraints of the single devices \eqref{eq:complete_power_limit_energy_focus_boundaries}, i.e.  
\begin{equation*}
\setevhr{} = \setevhr{1}  \times ... \times \setevhr{N} \subset \mathbb{R}^N. %\label{eq:sd_en_limit_vector}
\end{equation*}
The projection of set $\setevhr{}$  with $s$ from \eqref{eq:aggregating_function} is again to be understood as the point-wise image. This can be written as
\[
s(\setevhr{}) = \left[ \sum_{\mathclap{j\in \setEder}}\lb{e}{j}^r(k+1),~  \sum_{\mathclap{j\in \setEder}} \ub{e}{j}^r(k+1)\right].
\]
Next, we investigate the relation between the bounds of $\setevhragg$ and $\setevhstaragg$.
The application of \eqref{eq:agg_borders} to \eqref{eq:aggregated_bounds} implies the following inequalities 
\begin{subequations} \label{eq:joint_bounds_inequalities}
\begin{align}
%\max \left(\sum_{\mathcal{N}}\lb{e}{j}(k+1),\sum_{\mathcal{N}}\evh{j}(k) + \sum_{\mathcal{N}}\Dtk\lb{p}{j}(k)\right) \leq \sum_{\mathcal{N}} \max\left(\lb{e}{j}(k+1),\evh{j}(k) + \Dtk\lb{p}{j}(k)\right)
%\max\left(\lb{E}{}(k+1),E(k) + \Dtk\lb{P}{}(k)\right) \leq \sum_{\mathcal{N}} \max\left(\lb{e}{j}(k+1),\evh{j}(k) + \Dtk\lb{p}{j}(k)\right),\\
%\min\left(\ub{E}{}(k+1),E(k) + \Dtk\ub{P}{}(k)\right) \geq \sum_{\mathcal{N}} \min\left(\ub{e}{j}(k+1),\evh{j}(k) + \Dtk\ub{p}{j}(k)\right),
%&\lb{E}{}^r\rcekagg \leq \sum_{\mathclap{j\in \setEder}} \max\left(\lb{e}{j}(k+1),\evh{j}(k) + \Dtk\cdot\lb{p}{j}(k)\right),\\
%&\ub{E}{}^r\rcekagg \geq \sum_{\mathclap{j\in \setEder}} \min\left(\ub{e}{j}(k+1),\evh{j}(k) + \Dtk\cdot\ub{p}{j}(k)\right).
&\lb{E}{}^r\rcekagg \leq \sum_{\mathclap{j\in \setEder}}\lb{e}{j}^r(k+1),\\
&\ub{E}{}^r\rcekagg \geq \sum_{\mathclap{j\in \setEder}} \ub{e}{j}^r(k+1).
\end{align}
\end{subequations}
Consequently, we have that
\begin{equation}\label{eq:constraint_relaxed_large}
\setevhragg \supseteq \setevhstaragg.
\end{equation}
Specifically, there are values of $E(k+1)$ which are feasible for \eqref{eq:agg_pow_balance}-\eqref{eq:agg_en_limit} but not for \eqref{eq:constr_pow_balance_with_agg}-\eqref{eq:sum_energy_cons_with_agg} every time that one of the inequalities \eqref{eq:joint_bounds_inequalities} holds strictly.
Practically speaking, these are values of $E(k+1)$ that are feasible for the aggregated model but that cannot be dispersed into a feasible $\evect(k+1)$. 
This is the case of Example 1, illustrated in Fig. \ref{fig:set_example_intervals}c. 
It can be seen that the interval $\setevhstaragg$ is a ``narrower" subset of $\setevhragg$.
The $E(k+1)$ chosen in Example \ref{ex:inf_states} is feasible and reachable according to the aggregated model, as it is contained within $\setevhragg$.
However, this solution is unfeasible for the complete constraint set, because it lays \textit{outside} of the interval $\setevhstaragg$.

To summarize, the aggregated constraints from \eqref{eq:agg_borders} might enlarge the feasible space of $P(k)$ and $E(k)$.
Recalling the abstract problem of Section \ref{sec:abstract_problem}, this is the case if
\begin{equation}
\label{eq:general_set_inclusion}
\mathcal{Y}(k+1, \bold{y}(k)) \supseteq  h(\mathcal{X}(k+1, \bold{x}(k))),
\end{equation}
which violates condition \eqref{eq:general_set_equivalence} of Lemma \ref{lem:equiv}.
Subsequently we present sufficient conditions under which the aggregated constraints equal the projection of the original constraints.

\section{Ensuring Schedule Feasibility} \label{sec:feas_schedule}

From Lemma \ref{lem:equiv}, the aggregation preserves the optimality of the solution if and only if
\begin{equation} \label{eq:prop_agg_opt}
%\setevhragg \equiv \left[ \sum_{\mathclap{j\in \setEder}}\lb{e}{j}^r(k+1),\sum_{\mathclap{j\in \setEder}} \ub{e}{j}^r(k+1)\right],
\setevhragg \equiv \setevhstaragg.
\end{equation}
In the previous section we have showed that ${\setevhragg}$ obtained by application of \eqref{eq:agg_borders} does not guarantee \eqref{eq:prop_agg_opt}.
A natural consequence would be to find a different way to compute ${\setevhragg}$ such that \eqref{eq:prop_agg_opt} holds.
In contrast, we propose to pursue a reversed approach.
Instead of computing a ``narrower" ${\setevhragg}$ such that \eqref{eq:prop_agg_opt} is satisfied, we suppose the existence of an $\evect(k)$---whose elements sum up to the selected $E(k)$---for which \eqref{eq:prop_agg_opt} holds with ${\setevhragg}$ computed via \eqref{eq:agg_borders}. 
This is demonstrated in Example \ref{ex:f_states}. 
\begin{example} \label{ex:f_states}
Consider the setting of Example \ref{ex:inf_states}, but with $e_1(k) = e_2(k) = 0.5(\ub{e}{}+\lb{e}{})$. 
Fig. \ref{fig:set_example2} depicts this case.
Therein, $\setevhragg$ and ${\setevhstaragg}$ are equivalent and $E(k+1) = 2\ub{e}{}$---which is feasible and reachable for the aggregated model as discussed in Example \ref{ex:inf_states}---is feasible and reachable with respect to the constraints of each EC-DER: 
\begin{align*}
\pvhnb{1}(k) &= \ub{p}{} \in [\lb{p}{},\ub{p}{}], \qquad \pvhnb{2}(k) = \ub{p}{} \in [\lb{p}{},\ub{p}{}], \\
\evh{1}(k+1) &= 0.5(\ub{e}{}+\lb{e}{}) + \Dtk\cdot\ub{p}{} \in [\lb{e}{},\ub{e}{}],  \\ 
\evh{2}(k+1) &= 0.5(\ub{e}{}+\lb{e}{})+ \Dtk\cdot\ub{p}{} \in [\lb{e}{},\ub{e}{}].  \tag*{$\blacksquare$}% 
\end{align*} 
\end{example}

\begin{figure}
\centering
\includegraphics[width=0.45\textwidth]{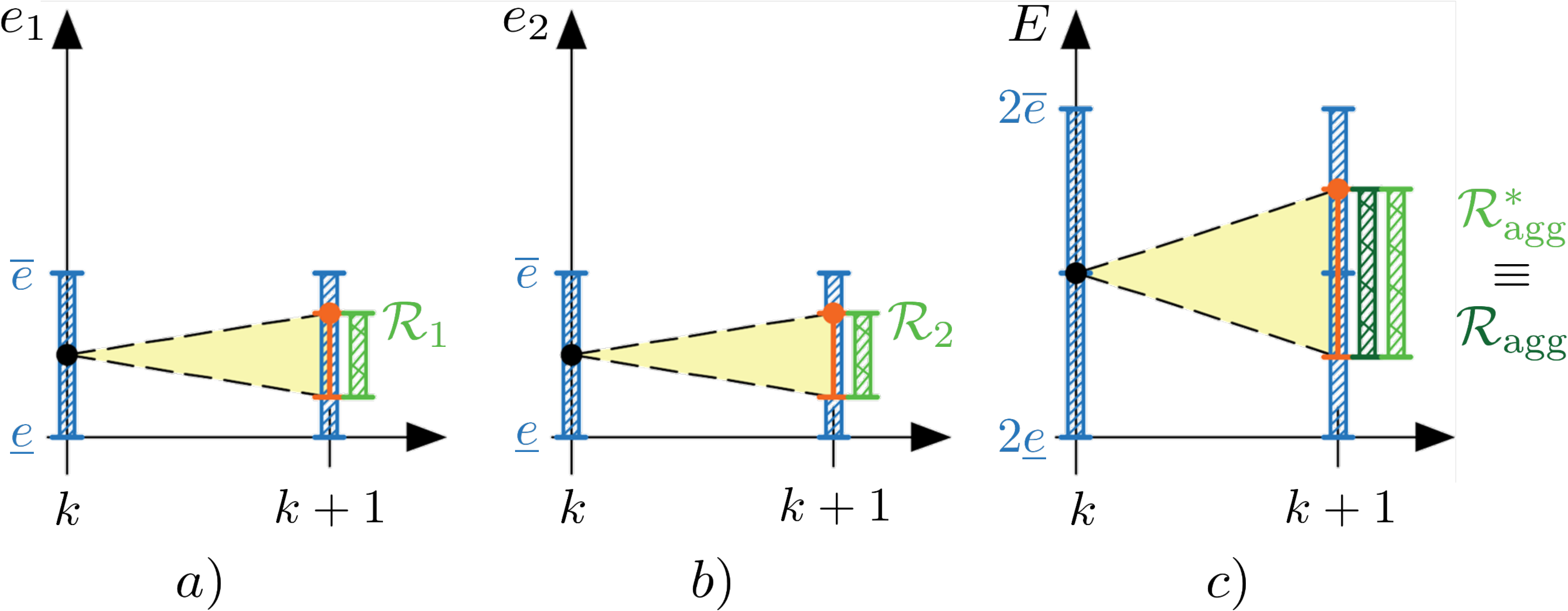}
\caption{Illustration of Example 2. \label{fig:set_example2} 
}
\vspace{-0.4cm}
\end{figure}

Comparison of Examples \ref{ex:inf_states} and \ref{ex:f_states} points out that certain feasible dispersions of $E(k)$ are preferable over others.
The energy states $e_1(k)$ and $e_2(k)$ are a feasible dispersion of aggregated energy $E(k)$ both in Example \ref{ex:inf_states} and in Example \ref{ex:f_states}. 
However, the values for $e_1(k)$ and $e_2(k)$ chosen in Example \ref{ex:inf_states} compromise the existence of a feasible dispersion of the desired (aggregated feasible) $E(k+1)$.
The same $E(k+1)$ has instead a feasible dispersion with the values of $e_1(k)$ and $e_2(k)$ chosen in Example \ref{ex:f_states}. 

Let us rephrase this observation using the abstraction from Section \ref{sec:abstract_problem}.
Given a value of $\bold{y}(k)$, we define % (with a slight abuse of notation)
\begin{equation*}
h^{-1}(\bold{y}(k))
%\mathcal{X}(k,\bold{y}(k)) 
= \left \lbrace \bold{x}(k) \in \mathcal{X}(k, \bold{x}(k-1)) \,|\, h(\bold{x}(k)) = \bold{y}(k) \right \rbrace
\end{equation*}
as the set of feasible pre-images of $\bold{y}(k)$ with respect to the projection $h$.
If the set $h^{-1}(\bold{y}(k))$ is not a singleton, each of its elements maps to $\bold{y}(k)$.
However, some of those can restrict the feasible set at the next time step, $\mathcal{X}(k+1, \bold{x}(k))$, more than others.
Hence, given a set $\mathcal{Y}(k+1, \bold{y}(k))$ for which \eqref{eq:general_set_inclusion} holds, satisfaction of \eqref{eq:general_set_equivalence} depends on $\bold{x}(k) \in h^{-1}(\bold{y}(k))$, because the set $h(\mathcal{X}(k+1, \bold{x}(k)))$ varies with $\bold{x}(k)$.  

Motivated by these considerations, we denote 
\begin{equation*}
\mathcal{E}(k,E(k)) \doteq \left\lbrace \evect(k) \in \setevh{}{k} \, |\, \evect(k) = s^{-1}(E(k)) \right\rbrace,
\end{equation*} 
as \emph{the set of feasible dispersions of $E(k)$}.
\begin{defi}[Consistent dispersion of $E(k)$] \label{def:cons_disp}
Given $E(k)$, the vector $\evect(k)$ is said to be \emph{a consistent dispersion of $E(k)$ at time $k+1$}, if $\evect(k) \in \mathcal{E}(k,E(k))$ and 
\begin{equation*}
\setevhragg \equiv \setevhstaragg.
\end{equation*}
with $\setevhragg$ as in \eqref{eq:aggregated_bounds} and $\setevhstaragg$ from \eqref{eq:direct_computation_of_agg}.
\QEDA
\end{defi}
The meaning of consistent $\evect(k)$ is illustrated in Fig. \ref{fig:consistentDispersion}.
Starting from the top-left corner, Fig. \ref{fig:consistentDispersion} shows that an aggregated state $E(k)$ can be dispersed in different ways, namely all the points in $\mathcal{E}(k,E(k))$.
Among them, a consistent dispersion of the aggregated state $E(k)$ leads to a feasible and reachable set $\setevhr{}$ (on the right) whose projection on the aggregated space, $\setevhstaragg$, is equivalent to the aggregated feasible and reachable set obtained via direct computation, i.e. from $E(k)$ and relative aggregated constraints as in \eqref{eq:agg_borders}.   
Considering the complete problem, the schedule $\{G(k)\}_{\mathcal{K}}$ can be computed via \eqref{eq:aggregated_opt_problem} as much as via \eqref{eq:complete_opt_problem_with_agg} if $\evect(k)$ is a consistent dispersion of $E(k)$ for all $k\in\mathcal{K}$, because the conditions of Lemma \ref{lem:equiv} are satisfied at each time step. 

\begin{figure}
\centering
\includegraphics[width=0.46\textwidth]{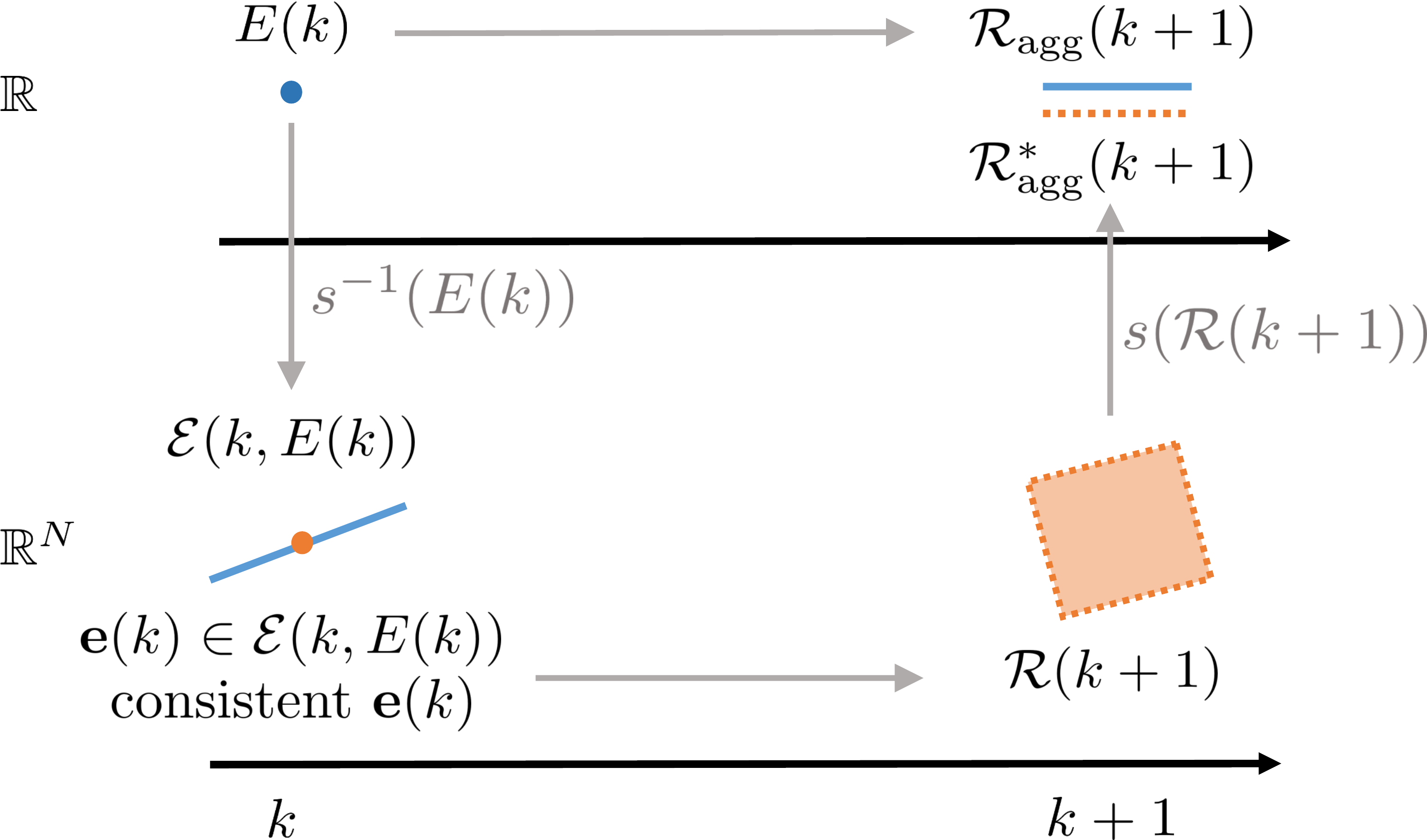}
\caption{Graphical representation of a consistent dispersion of $E(k)$.
\label{fig:consistentDispersion} 
}
\vspace{-0.4cm}
\end{figure}

\begin{rema}[Dispersions and two-stage scheduling]
Introducing the concept of a consistent dispersion helps understanding why \eqref{eq:aggregated_opt_problem} is successfully applied, cf. \cite{Xu16,Appino18c}.
Therein, hierarchical control is often used to solve \eqref{eq:complete_opt_problem_with_agg} in a sequential manner.
First, \eqref{eq:aggregated_opt_problem} is used (with \eqref{eq:agg_borders}) on an upper level to provide long-term aggregated schedules. 
Second, lower level feedback control (e.g. MPC) assigns the dispersion of the aggregated energy state in a way that guarantees the feasibility of the pre-computed schedule over the subsequent time-steps. 
Essentially, this controller consistently disperses the aggregated schedule.\QEDA
\end{rema}

\section{Existence and Computation of Consistent Dispersions} \label{sec:prove_agg_opt}

The considerations above motivate investigation of the existence of a consistent dispersion.
To this end, we make the following assumption.

\begin{ass} \label{req:meaningful_device_constraints_part2}
For all $j \in \mathcal{N}$  and all $k \in \mathcal{K}$, let the sampling time $\delta$ and the intervals $\setpvhnb{j}{k}$, $\setevh{j}{k}$ and $\setevh{j}{k+1}$ be such that 
there exists at least one $\evh{j}(k) \in \setevh{j}{k}$ for which
\begin{subequations}
\begin{align}
\evh{j}(k) \oplus \Dtk  \cdot \setpvhnb{j}{k} &\subseteq \setevh{j}{k+1}, \label{eq:ass2_1}\\
\setevh{j}{k+1} &\subseteq \left( \setevh{j}{k} \oplus \Dtk \cdot \setpvhnb{j}{k} \right), \label{eq:ass2_2}
\end{align}
\end{subequations}
holds. \QEDA
\end{ass}
Condition \eqref{eq:ass2_1} implies that there exists a controlled invariant (time-varying) subset of $\setevh{j}{k}$, and furthermore that there exists a subset of $\setevh{j}{k}$ from which all reachable states are  feasible.
The second part \eqref{eq:ass2_2} instead requires that the entire set $\setevh{j}{k+1}$ is contained in the reachable set $\setevh{j}{k} \oplus\Dtk \cdot \setpvhnb{j}{k}$. 

Moreover, consider
$\costDispLow: \mathbb{R} \times \mathbb{R}^N \to \mathbb{R}$  given by
\begin{subequations}
\label{eq:cost_function_ideal_state} \begin{multline}
\costDispLow(k,\evect(k))  = \sum_{j \in \setEder} - \bigg(\evh{j}(k) + \Dtk \cdot \underline{p}_j(k) +  \\
  \quad - \max\left\{\evh{j}(k) +{\Dtk} \cdot \underline{p}_j(k),\, \underline{e}_j(k+1) \right\}\bigg), 
 \end{multline}
and $\costDispUp:\mathbb{R} \times \mathbb{R}^N \to \mathbb{R}$ given by
 \begin{multline}
\costDispUp(k,\evect(k)) = \sum_{j \in \setEder} \bigg(\evh{j}(k) + \Dtk \cdot \overline{p}_j(k) +  \\
   \quad - \min\left\{\evh{j}(k) + {\Dtk} \cdot \overline{p}_j(k),\, \overline{e}_j(k+1)\right\}\bigg). 
\end{multline}
\end{subequations}

\begin{thm}[Existence of a consistent dispersion] \label{thm:theorem_1}~\\
Suppose Assumptions \ref{req:meaningful_device_constraints_part1} and \ref{req:meaningful_device_constraints_part2} hold. If
\begin{equation} \label{eq:problem_consistent_soc}
\evect(k) \in \argmin_{\evect(k) \in \mathcal{E}(k,E(k))} \left( \costDispLow(k,\evect(k)) + \costDispUp(k,\evect(k)) \right)
\end{equation}
with $\costDispLow$ and $\costDispUp$ from \eqref{eq:cost_function_ideal_state}, 
then $\setevhragg$ from \eqref{eq:aggregated_bounds} and $\setevhstaragg$ from \eqref{eq:direct_computation_of_agg} satisfy
\[
\setevhragg \equiv \setevhstaragg,
\]
 i.e. $\evect(k)$ is a consistent dispersion of $E(k)$ at time $k+1$.\QEDA
\end{thm}

Note that Assumption \ref{req:meaningful_device_constraints_part2} is quite mild. 
The left hand side set relation can, for example, be enforced by reducing the sampling period $\delta$. 
The right hand side relation requires exclusion of all the values of $\evect(k+1)$ which are feasible from an energy perspective but that cannot be reached from any feasible $\evect(k)$.
Similar to the case of Assumption \ref{req:meaningful_device_constraints_part1}, this latter aspect is an additional modeling effort that is normally avoided, but that is fundamental in case of aggregation.  

Furthermore, observe that cost function $\costDispLow(k,\evect(k))+\costDispUp(k,\evect(k))$ as stated in Theorem \ref{thm:theorem_1} involves the energy state $\evect(k)$ \textit{only} at time step $k$.
This means that the computation of a consistent dispersion is independent from the knowledge of $\evect(h)$ with $h \neq k$.
Thus, $\evect(k)$ (and consequently $\pvect(k-1)$) can be determined at each $k$ by a lower-level controller without any need for information about past or future states of the system.
This has two important consequences: 
On one hand, it avoids any dependency between dispersion at subsequent time instants in the scheduling problem, which complicates dealing with eventual random decision variables and parameters. 
On the other hand, the computation of $\evect(k)$ does not require any load forecasts nor MPC strategies (the long-term perspective is already accounted for in the aggregated schedule).  
Theorem \ref{thm:theorem_1} leads to the implication of interest. 

\begin{thm}[Recursive existence of consistent dispersions] \label{thm:existence_of_ideal_states}
Suppose Assumptions \ref{req:meaningful_device_constraints_part1} and \ref{req:meaningful_device_constraints_part2} hold.
If $E(k+1)\in \setevhagg{k+1}$ and $\evect(k)$ is a consistent dispersion of $E(k)$, then there exists at least one consistent dispersion of $E(k+1)$.\QEDA
\end{thm}

An important consequence of Theorem \ref{thm:existence_of_ideal_states} is that given that the initial $\evect(0)$ is a consistent dispersion of $E(1)$, then a consistent dispersion of $\{P(k)\}_{\mathcal{K}}$ and $\{E(k+1)\}_{\mathcal{K}}$ is always possible as long as \eqref{eq:agg_pow_limit}-\eqref{eq:agg_en_limit} hold with \eqref{eq:agg_borders}.
In other words, the aggregated schedule can always be tracked. 
Furthermore, note that while the existence of a consistent dispersion is a property of the entire system, the requirements of Assumptions \ref{req:meaningful_device_constraints_part1} and \ref{req:meaningful_device_constraints_part2} involve each EC-DER \textit{separately}, in line with the task of aggregating heterogeneous devices. 

Besides justifying the use of aggregated scheduling \eqref{eq:aggregated_opt_problem} in presence of a consistent dispersion of the aggregated energy state, Theorem \ref{thm:theorem_1} and Theorem \ref{thm:existence_of_ideal_states} lead to further considerations in more general problems. 
We will discuss this in Section \ref{sec:limits_and_adv_of_agg_model} after providing the proofs in the next Section. 

\subsection{Proofs of Theorem \ref{thm:theorem_1} \& \ref{thm:existence_of_ideal_states}}
\begin{proof}[Proof of  Theorem \ref{thm:theorem_1}]~\\
The main idea behind the proof of  Theorem \ref{thm:theorem_1} is showing that  condition $\setevhragg \equiv \setevhstaragg$ is satisfied because the bounds of the real intervals $\setevhragg$ and $\setevhstaragg$ coincide when $\evect(k)$ satisfies \eqref{eq:problem_consistent_soc}. 
First we prove technical lemmata to then turn towards the proof of Theorem \ref{thm:theorem_1}. 
 
From \eqref{eq:aggregated_bounds} it follows that there are two possible cases for each of the two bounds of $\setevhragg$:
\begin{itemize}
	\item Case (i):  
 	\begin{subequations}
 		\begin{equation} E(k) + \Dtk \cdot \underline{P}(k) \geq {\underline{E}(k+1)}, \label{eq:case1_low} \end{equation}
 	respectively, 
   		\begin{equation} E(k) + {\Dtk} \cdot \overline{P}(k) \leq {\overline{E}(k+1)}. \label{eq:case1_up} \end{equation}
    \end{subequations}
 	\item Case (ii):  
 	\begin{subequations}
 		\begin{equation} E(k) + {\Dtk} \cdot \overline{P}(k) > {\overline{E}(k+1)}, \label{eq:case2_low}  \end{equation}
 	respectively,  
 		\begin{equation}E(k) + \Dtk \cdot \underline{P}(k) < {\underline{E}(k+1)}.  \label{eq:case2_up}\end{equation}
   	\end{subequations}
\end{itemize}
Observe that Case (i) implies that at time $k$ the aggregated input (power) constraint \eqref{eq:agg_pow_limit} is restrictive (i.e. potentially active), while Case (ii) implies that at time $k+1$ the aggregated state (energy) constraint  \eqref{eq:agg_en_limit} is restrictive. 
We analyze both cases in two technical lemmata. 
\begin{lemm}[Case (i)] \label{lem:case_i}
Suppose Assumptions \ref{req:meaningful_device_constraints_part1} and \ref{req:meaningful_device_constraints_part2} hold and let $\evect(k)$ satisfy \eqref{eq:problem_consistent_soc}.
\begin{enumerate}[label=(\alph*)]
\item 	If \eqref{eq:case1_low} holds, 
		then the lower bound $\lb{E}{}^r\rcekagg$ of  $\setevhragg$ from \eqref{eq:aggregated_bounds} and the lower bound ${\sum_{{j\in \setEder}}\lb{e}{j}^r(k+1)}$ of $\setevhstaragg$ from 		\eqref{eq:direct_computation_of_agg} satisfy 
		 \begin{equation*}
		 \lb{E}{}^r\rcekagg = \sum_{\mathclap{j\in \setEder}}\lb{e}{j}^r(k+1). 
		 \end{equation*}
\item If \eqref{eq:case1_up} holds,
		then the upper bound $\ub{E}{}^r\rcekagg$ of  $\setevhragg$ from \eqref{eq:aggregated_bounds} and the upper bound ${\sum_{{j\in \setEder}}\ub{e}{j}^r(k+1)}$ of $\setevhstaragg$ from 		\eqref{eq:direct_computation_of_agg} satisfy 
		 \begin{equation*}
		 \ub{E}{}^r\rcekagg = \sum_{\mathclap{j\in \setEder}}\ub{e}{j}^r(k+1). 
		 \end{equation*}
\end{enumerate}
\end{lemm}
\begin{proof}
First, consider the lower-bound case (a), i.e. that \eqref{eq:case1_low} holds. 
From $ E(k) + \Dtk \cdot \underline{P}(k) \geq {\underline{E}(k+1)}$, we have
\[
E(k) \geq - \Dtk \cdot \underline{P}(k) + \underline{E}(k+1).
\] 
Given Assumption \ref{req:meaningful_device_constraints_part1} this implies that any $\evect(k)$ satisfying
\begin{equation} \label{eq:conditions_min_case1_lower}
{e_j(k) \geq - \Dtk \cdot \underline{p}_j(k) + \underline{e}_j(k+1)} \quad \forall j \in \setEder,
\end{equation}
is a feasible dispersion of $E(k)$; i.e. $\evect(k)$ is in the set $\mathcal{E}(k,E(k))$.
 
Observe that $\costDispLow(k,\evect(k))$ reaches its global minimum zero for any $\evect(k)$ satisfying \eqref{eq:conditions_min_case1_lower}.
Furthermore, from Assumption \ref{req:meaningful_device_constraints_part2} it holds that
\begin{equation} \label{eq:proof_consequence_ass2}
\ub{e}{j}(k+1) - \Dtk \cdot \ub{p}{j}(k) \geq \lb{e}{j}(k+1) - \Dtk \cdot \lb{p}{j}(k).
\end{equation}
This implies that diminishing the values of any $e_j(k)$ below ${-\Dtk \cdot \underline{p}_j(k) + \underline{e}_j(k+1)}$ does not reduce the value of $\costDispUp(\evect(k),k)$. 
Therefore, \eqref{eq:problem_consistent_soc} for Case (i)-(a) implies \eqref{eq:conditions_min_case1_lower}.

In turn, it follows from \eqref{eq:conditions_min_case1_lower} that
\[
{e_j(k) + \Dtk \cdot \underline{p}_j(k) \geq \underline{e}_j(k+1)} \quad \forall j \in \setEder.
\]
Therefore, the lower bound of $\setevhstaragg$ is equal to 
\[
\sum_{\mathclap{j\in \setEder}}\lb{e}{j}^r(k+1) = \sum_{\mathclap{j\in \setEder}} \evh{j}(k) + \Dtk \cdot \sum_{\mathclap{j\in \setEder}} \lb{p}{j}(k) = \lb{E}{}^r\rcekagg.
\]
The second part of Case (i)---$E(k) + {\Dtk} \cdot \overline{P}(k) \leq {\overline{E}(k+1)}$---follows mutatis mutandis and it is skipped for the sake of brevity. 
\end{proof}
\begin{lemm}[Case (ii)]\label{lem:case_ii}
Suppose Assumptions \ref{req:meaningful_device_constraints_part1} and \ref{req:meaningful_device_constraints_part2} hold and let $\evect(k)$ satisfy \eqref{eq:problem_consistent_soc}.
\begin{enumerate}[label=(\alph*)]
\item If \eqref{eq:case2_low} holds, 
	then the lower bound $\lb{E}{}^r\rcekagg$ of  $\setevhragg$ from \eqref{eq:aggregated_bounds} and the lower bound ${\sum_{{j\in \setEder}}\lb{e}{j}^r(k+1)}$ of $\setevhstaragg$ from 	\eqref{eq:direct_computation_of_agg} satisfy 
	\begin{equation*}
 		\lb{E}{}^r\rcekagg = \sum_{\mathclap{j\in \setEder}}\lb{e}{j}^r(k+1). 
 	\end{equation*}
\item If \eqref{eq:case2_up} holds,
	then the upper bound $\ub{E}{}^r\rcekagg$ of  $\setevhragg$ from \eqref{eq:aggregated_bounds} and the upper bound ${\sum_{{j\in \setEder}}\ub{e}{j}^r(k+1)}$ of $\setevhstaragg$ from 		\eqref{eq:direct_computation_of_agg} satisfy 
		 \begin{equation*}
		 \ub{E}{}^r\rcekagg = \sum_{\mathclap{j\in \setEder}}\ub{e}{j}^r(k+1). 
		 \end{equation*}
\end{enumerate} 
 \end{lemm}
 \begin{proof}[Proof (by contradiction)]
First, consider the lower-bound case (b), i.e. that \eqref{eq:case2_low} holds.  
Assume for the sake of contradiction that there is a group of devices $\mathcal{I} \subset \setEder$ for which 
\begin{equation} \label{eq:proof_Idevice_inequality}
e_i(k) + \Dtk \cdot \underline{p}_i(k) > {\underline{e}_i(k+1)}, \quad \forall i \in\mathcal{I}
\end{equation} 
and another group of devices $\mathcal{H} = \setEder\setminus\mathcal{I}$ for which 
\begin{equation} \label{eq:proof_Hdevice_inequality}
e_h(k) + \Dtk \cdot \underline{p}_h(k) \leq {\underline{e}_h(k+1)},\quad \forall h\in \mathcal{H}.
\end{equation}
From \eqref{eq:agg_powers} it follows that 
\begin{equation*}
E(k) = \sum_{\mathcal{I}}e_i(k) + \sum_{\mathcal{H}}e_h(k).
\end{equation*}

Suppose that $\mathcal{H}= \emptyset$, then \eqref{eq:case2_low} is false and we fall back to Case (i). 
Hence  $\mathcal{H}\neq \emptyset$. 
Moreover, for at least one $\tilde h \in \mathcal{H}$, Inequality \eqref{eq:proof_Hdevice_inequality} holds strictly (otherwise we again fall back to Case (i)).

Consider a small positive increment $\Delta > 0$ for which  $\tilde{e}_{\tilde h}(k) = {e}_{\tilde h}(k) + \Delta$  still satisfies \eqref{eq:proof_Hdevice_inequality}. 
The increment of $\tilde{e}_{\tilde h}(k)$ is always allowed. 
In fact, given Assumption \ref{req:meaningful_device_constraints_part2}, we have 
\[
\overline{e}_{\tilde h}(k) + \Dtk \cdot \underline{p}_{\tilde h}(k) \geq \underline{e}_{\tilde h}(k+1).
\]
Thus, from \eqref{eq:proof_Hdevice_inequality} (which is strictly satisfied for ${\tilde h}$), it follows that $e_{\tilde{h}}(k) < \overline{e}_{\tilde{h}}(k)$.
Next, consider w.l.o.g. one EC-DER $\tilde{i} \in \mathcal{I}$ for which the energy $\tilde{e}_{\tilde{i}}(k) = {e}_{\tilde{i}}(k) - \Delta$ is decreased such that $\tilde{e}_{\tilde{i}}(k)$ still satisfies \eqref{eq:proof_Idevice_inequality}.
This reduction is always possible:
given Assumption \ref{req:meaningful_device_constraints_part2} we have
\[
\underline{e}_{\tilde i}(k) + \Dtk \cdot \underline{p}_{\tilde i}(k) \leq \underline{e}_{\tilde i}(k+1);
\]
thus, \eqref{eq:proof_Idevice_inequality} implies $e_{\tilde i}(k) > \underline{e}_{\tilde i}(k)$.  

Increasing the energy at EC-DER $\tilde{h}$ and reducing it at the same time at EC-DER $\tilde{i}$ maintains the satisfaction of equality \eqref{eq:agg_powers}, as 
\[
E(k) = - \Delta +\sum_{\mathcal{I}}e_i(k) + \sum_{\mathcal{H}}e_h(k) + \Delta.
\] 
However, $\costDispLow(\tilde{\evect}(k),k) < \costDispLow({\evect}(k),k)$.
Moreover, $\costDispUp(\tilde{\evect}(k),k) \leq \costDispUp({\evect}(k),k)$.
In fact, the increment of $e_{\tilde{h}}(k)$ cannot increase $\costDispUp({\evect}(k),k)$, because the contribution of the components $e_h(k)$ to $\costDispUp({\evect}(k),k)$ is  zero.\footnote{Given \eqref{eq:proof_consequence_ass2}, then $\tilde{e}_{\tilde{h}}(k) \leq \ub{e}{\tilde{h}}(k+1) - \Dtk \cdot \ub{p}{\tilde{h}}(k)$.}
Thus, $\evect(k)$ does not minimize $\costDispLow(\evect(k),k)+\costDispUp({\evect}(k),k)$. 
We arrive at a contradiction, i.e. $\mathcal{I}\neq \emptyset$ contradicts \eqref{eq:problem_consistent_soc}. Hence we arrive at $\mathcal{I} = \emptyset$.

Given that \eqref{eq:proof_Hdevice_inequality} hold for all the EC-DERs, then the lower bound of $\setevhstaragg$ is equal to
\[
\sum_{\mathclap{j\in \setEder}}\lb{e}{j}^r(k+1) = \sum_{\mathclap{j\in \setEder}} \lb{e}{j}(k+1) = \lb{E}{}^r\rcekagg. 
\] 
Summing up, the lower bounds  of $\setevhragg$ and  $\setevhstaragg$ also match in Case (ii)-(a) under \eqref{eq:problem_consistent_soc}.
In a similar fashion, it can be shown that the upper bound of $\setevhragg$ equals the upper bound of $\setevhstaragg$ in Case (ii)-(b) under \eqref{eq:problem_consistent_soc}.
 \end{proof}
 As we have seen in Lemma \ref{lem:case_i} for  Case (i) and in Lemma \ref{lem:case_ii} for Case (ii), if \eqref{eq:problem_consistent_soc} holds then the upper and lower bound of $\setevhragg$ and $\setevhstaragg$ coincide, meaning that $\setevhragg \equiv \setevhstaragg$. This finishes the proof of Theorem \ref{thm:theorem_1}. 
\end{proof}

\begin{proof}[Proof of Theorem \ref{thm:existence_of_ideal_states}]~\\
By construction, if $\evect(k)$ is a consistent dispersion of $E(k)$ and $E(k+1)\in \setevhagg{k+1}$, then there exist at least a feasible dispersion of $E(k+1)$, cf. Definition \ref{def:cons_disp}.
Thus, the set ${\mathcal{E}(k+1,E(k+1))}$ is not empty.
The cost function $\costDispLow(k+1,\evect(k+1)) + \costDispUp(k+1,\evect(k+1))$ is real-valued and continuous and the set ${\mathcal{E}}(k+1,E(k+1))$ is compact.
Thus, by virtue of the extreme value theorem, the function $\costDispLow(k+1,\evect(k+1)) + \costDispUp(k+1,\evect(k+1))$ attains a minimum over ${\mathcal{E}}(k+1,E(k+1))$ meaning that at least one minimizer $\evect(k+1)$ exists.
Consequently, applying Theorem \ref{thm:theorem_1}, there exist a consistent dispersion of $E(k+1)$.
\end{proof}

\section{Discussion}
\label{sec:limits_and_adv_of_agg_model}

Theorems \ref{thm:theorem_1} and \ref{thm:existence_of_ideal_states} indicate the following approach to scheduling problems in the form of \eqref{eq:complete_opt_problem}:
\begin{enumerate}[label=(\roman*)]
\item Check/impose consistency of constraints \eqref{eq:sd_state_eq_det_vector}-\eqref{eq:sd_en_limit_vector} to Assumptions \ref{req:meaningful_device_constraints_part1} and \ref{req:meaningful_device_constraints_part2}.
\item Aggregate decision variables and constraint sets via projection $s$, as in \eqref{eq:agg_powers} and \eqref{eq:agg_borders}.
\item Solve \eqref{eq:aggregated_opt_problem} to compute the schedule $\{G(k)\}_{\mathcal{K}}$;
\item Guarantee feasibility of the aggregated schedule by dispersing the aggregated variables following \eqref{eq:problem_consistent_soc}. 
\end{enumerate}
   
However, it can be argued that the class of problems covered by \eqref{eq:complete_opt_problem} does not include several cases of practical interest.
In particular, problem \eqref{eq:complete_opt_problem} A) disregards non-linearities in the EC-DERs dynamics  \eqref{eq:sd_state_eq_det_vector}, and B) ignores constraints that involve a partial summation of the elements of $\pvect(k)$ and $\evect(k)$.
In power systems, non-linearities in the EC-DERs dynamics are used to model conversion losses, while network constraints comprise partial summations.
Next we discuss Issue A), while Issue B) is left out due to space limitations.  

\subsection{Conversion Losses} \label{sec:system_losses}
We seek a constraint reformulation that allows using Theorems \ref{thm:theorem_1} and \ref{thm:existence_of_ideal_states} for an aggregation of nonlinear dynamics.

First, consider a dynamic model of the energy state including conversion losses, for example 
\begin{equation}
\label{eq:sd_state_eq_det_loss}
\evh{j}(k+1) = \evh{j}(k) + \Dtk \cdot \pvhnb{j}(k) - \mu_j(\pvhnb{j}(k)) \Dtk \cdot \left| \pvhnb{j}(k) \right|,
\end{equation}
where the coefficient $\mu_j(\pvhnb{j}(k))$ describes the amount of energy that is lost in conversion. 
Note that this coefficient can be a function of $\pvhnb{j}(k)$, see \cite{Barry17}.
The dynamics \eqref{eq:sd_state_eq_det_loss} contain an energy-based description of the storage losses.
This model is often applied in energy management even if it simplifies the phenomena occurring in practice; nonetheless, it has been established experimentally to be sufficiently accurate for scheduling purposes, cf. \cite{Sossan16a}. 

Next, consider an alternative---yet equivalent---model describing the conversion losses in terms of lost power.
First, we define 
\begin{equation}
\label{eq:power_into_storage}
\tilde{p}_j(k) = \pvhnb{j}(k) - \mu_j(\pvhnb{j}(k)) \cdot \left| \pvhnb{j}(k) \right|,
\end{equation}
as the actual power exchange with the storage.
Then, we introduce $\tilde{p}_j(k)$ into \eqref{eq:sd_state_eq_det_loss}, which gives
\begin{equation}
\label{eq:sd_state_eq_det_loss_power}
\evh{j}(k+1) = \evh{j}(k) + \Dtk \cdot \tilde{p}_j(k).
\end{equation}
Given \eqref{eq:power_into_storage}, the power constraint \eqref{eq:sd_pow_limit} can be equivalently expressed as a constraint on $\tilde{p}_j(k)$, i.e.
\begin{equation}
\label{eq:power_constraint_losses}
\tildesetpvect{k} = \left\lbrace \tildepvect(k) \, | \, \exists \, \pvect(k) \in \setpvect{k} \textrm{ s.t. \eqref{eq:power_into_storage} holds}\right\rbrace.
%  - \mu_j(\pvhnb{j}(k)) \Dtk \cdot \left| \pvhnb{j}(k) \right| \, \text{s.t.} \, \pvhnb{j}(k) \in \setpvect{k} \right\rbrace.
\end{equation} 
In other words, conversion losses modify the boundaries of the set of reachable states at $k+1$.
Note that $\tildesetpvect{k}$ is a closed real interval. 

Finally, consider the scheduling problem \eqref{eq:complete_opt_problem_with_agg} including conversion losses as modeled by \eqref{eq:power_into_storage}-\eqref{eq:power_constraint_losses},
\begin{subequations}\label{eq:complete_opt_problem_with_agg_with_losses}
\begin{align}
\min_{
\begin{subarray}{c}
  \{G(k)\}_{\mathcal{K}}, \\ \{\tilde{P}(k)\}_{\mathcal{K}},\{E(k+1)\}_{\mathcal{K}}, \\ \{\tildepvect(k)\}_{\mathcal{K}}, \{\evect(k+1)\}_{\mathcal{K}} % \forall	k\in\mathcal{K}
  \end{subarray}} &\sum_{k \in \mathcal{K}} c\left(G(k)\right)\\
  \text{s. t. \,\,} \forall k \in \mathcal{K} \quad& \nonumber \\
G(k)\phantom{+1,} =& \, \tilde{P}(k) - L(k) - P_{\text{loss}}(k), \label{eq:constr_pow_balance_with_agg_with_losses}\\ 
\evect(k+1) = & \, \evect(k) + \Dtk \cdot \tildepvect(k) \quad \evect(0) = \evect^0, \label{eq:sd_dynamic_vector_with_agg_with_losses} \\
\tildepvect(k)\phantom{+1,} \in& \, \tildesetpvect{k} , \label{eq:sd_pow_limit_vector_with_agg_with_losses} \\
\evect(k+1) \in& \,  \setevect{k+1}, \label{eq:sd_en_limit_vector_with_agg_with_losses} \\
\tilde{P}(k)\phantom{+1,} =& \, s(\tildepvect(k)),  \\
{E}(k+1) =& \, s(\evect(k+1)), \label{eq:sum_energy_cons_with_agg_with_losses} \\
P_{\text{loss}}(k)\phantom{1} =& f(\tildepvect(k)). \label{eq:power_loss_with_agg_with_losses}
%P_{\text{loss}}(k)\phantom{1} =& \, \mu_1(\pvhnb{1}(k)) \cdot \left| \pvhnb{1}(k) \right| + \hdots \hspace{-0.2cm} & \nonumber \\
%\phantom{P_{\text{loss}}(k)1, =}& \,  + \mu_N(\pvhnb{N}(k)) \cdot \left| \pvhnb{N}(k) \right| \hspace{-0.2cm} &\forall k \in \mathcal{K}. \label{eq:sum_power_loss_with_agg_with_losses}
%& \pvhnb{j}(k) \in  \setpvhnb{j}{k} \quad \forall j \in \mathcal{N}, \label{eq:sd_pow_limit} \\
%& \evh{j}(k+1) \in \setevh{j}{k+1} \quad \forall j \in \mathcal{N}, \label{eq:sd_en_limit}  
\end{align}
\end{subequations}
The function $f(\tildepvect(k))$ returns the total conversion losses (in terms of power) resulting from $\tildepvect(k)$.\footnote{Note that depending on the relation between $\mu_j(\pvhnb{j}(k))$ and $\pvhnb{j}(k)$, $f_{\text{loss}}(\tildepvect(k))$ may not even have a closed form description.}
Equation \eqref{eq:power_loss_with_agg_with_losses} is the only practical difference between problem \eqref{eq:complete_opt_problem_with_agg_with_losses} and \eqref{eq:complete_opt_problem_with_agg}.
Thus, the question arises if the solution of \eqref{eq:complete_opt_problem_with_agg_with_losses} can be approached in two steps (aggregated scheduling and subsequent dispersion), as for \eqref{eq:complete_opt_problem_with_agg}, without compromising the result. 

The aggregated version of \eqref{eq:complete_opt_problem_with_agg_with_losses} would be
\begin{subequations}\label{eq:aggregated_opt_problem_with_losses}
\begin{align}
\min_{
\begin{subarray}{c}
   \{G(k)\}_{\mathcal{K}}, \\ \{\tilde{P}(k)\}_{\mathcal{K}},  \{E(k+1)\}_{\mathcal{K}}
  \end{subarray}} &\sum_{k \in \mathcal{K}} c\left(G(k)\right)\\
  \text{s. t. \,\,} \forall k \in \mathcal{K} \quad& \nonumber \\
G(k)  =& \, \tilde{P}(k) - L(k) - f_{\text{agg}}(\tilde{P}(k)), \label{eq:agg_pow_balance_with_losses}\\ 
E(k+1) =& \, E(k) + \Dtk \cdot \tilde{P}(k), \quad E(0) = s(\evect^0), \label{eq:agg_state_eq_det_with_losses} \\
\tilde{P}(k) \phantom{+1,} \in& \,   \tildesetpvhnbagg{k}, \label{eq:agg_pow_limit_with_losses} \\
E(k+1) \in& \, \setevhagg{k+1}. \label{eq:agg_en_limit_with_losses} 
\end{align}
\end{subequations}
Observe that constraints \eqref{eq:sd_dynamic_vector_with_agg_with_losses}-\eqref{eq:sum_energy_cons_with_agg_with_losses} and \eqref{eq:agg_state_eq_det_with_losses}-\eqref{eq:agg_en_limit_with_losses} have the exact same structure of the corresponding ones in \eqref{eq:complete_opt_problem_with_agg} and \eqref{eq:aggregated_opt_problem}. 
Therefore, the idea of a consistent dispersion and the statements of Theorem \ref{thm:theorem_1} and \ref{thm:existence_of_ideal_states} still hold.
The main difficulty is represented here by enforcement of the equivalence between \eqref{eq:agg_pow_balance_with_losses} and \eqref{eq:constr_pow_balance_with_agg_with_losses}, \eqref{eq:power_loss_with_agg_with_losses}. 

Consider a symmetric system where the losses are the same on each device and independent of $\pvhnb{j}(k)$, i.e. $\mu_j(\pvhnb{j}(k))=\mu$ for all $j\in \mathcal{N}$.
Furthermore, consider the absence of mutual exchange of power among devices, i.e. $\text{sign}(\pvhnb{j}(k))=\text{sign}(\pvhnb{i}(k))$ for all $\{j,i\}\in \mathcal{N}\times \mathcal{N}$.
Then, defining
\begin{equation} \label{eq:aggrgeated_losses_calc}
f_{\text{agg}}(\tilde{P}(k)) =  \left\{\begin{matrix}
\frac{\mu}{1 - \mu} & \text{if } \tilde{P}(k) \geq 0\\ 
\frac{-\mu}{1 + \mu} & \text{if } \tilde{P}(k) < 0
\end{matrix}\right. ,
\end{equation}
the equivalence $f_{\text{agg}}(\tilde{P}(k)) = f(\tildepvect(k))$ holds exactly for each $\tilde{P}(k) \in   \tildesetpvhnbagg{k}$.  
Thus, in this specific case, \eqref{eq:agg_pow_balance_with_losses} is equivalent to \eqref{eq:constr_pow_balance_with_agg_with_losses}, \eqref{eq:power_loss_with_agg_with_losses}, and aggregation as in \eqref{eq:aggregated_opt_problem_with_losses} leads to the same solution of \eqref{eq:complete_opt_problem_with_agg_with_losses}.
However, in general $f_{\text{agg}}(\tilde{P}(k))$ can only approximate $P_{\text{loss}}(k)$ and the solution of \eqref{eq:aggregated_opt_problem_with_losses} deviates from the one of the original problem \eqref{eq:complete_opt_problem_with_agg_with_losses} depending on the severity of this approximation.
This aspect might compromise the effectiveness of aggregation. 
Eventually, aggregation of groups of EC-DERs with similar conversion losses can reduce the severity of the approximation at the price of a slightly higher number of variables in the aggregated problem. 

\section{Conclusion}
\label{sec:conclusion}

The present paper investigated scheduling for aggregated energy systems based on an aggregated time-varying model of heterogeneous energy-constrained distributed energy resources.
The aggregated model reduces the number of parameters and decision variables in the scheduling problem, lowering the computation burden and improving the management of the uncertainties.   
However, aggregated models, as often employed in the literature, may extend the feasible space to values that cannot be attained in practice.
This paper examined the causes of this undesirable phenomenon.
We show that, whenever the energy states of the individual devices satisfy a specific (collective) property---existence of a \textit{consistent dispersion of the aggregated energy state} at the first time-step---the aggregation does not  alter the feasible set and thus does not imply any optimality loss.
Our main contribution is to prove that a consistent dispersion always exists under mild assumptions on the constraints of the various devices, considered separately.
Our findings justify the application of an aggregated model in scheduling, provide a technique to check the consistency of an aggregated model, and propose a way to disperse the aggregated energy state which is free from any time dependence.
Furthermore, our results allow consideration on advantages and limitations of aggregated models in scheduling; we discussed here the case of conversion losses to provide an example.
Future work will apply the described methods to a real test case and further analyze the presence of network constraints.         

% or


%\appendices
%% you can choose not to have a title for an appendix
%% if you want by leaving the argument blank
%\section{}

%\section*{Acknowledgment}

\ifCLASSOPTIONcaptionsoff
  \newpage
\fi



%\bibliographystyle{IEEEtran}

\bibliographystyle{IEEEtran}
\bibliography{Appino_bib}{}
\end{document}